\documentclass[10pt,a4paper]{article}
\usepackage[english]{babel}
\usepackage{german}
\usepackage{amssymb}
\usepackage{amsfonts}
\usepackage{amsmath}
\usepackage[latin1]{inputenc}
\usepackage{fullpage}
\usepackage{cite}
\usepackage{graphicx}
\usepackage{subfigure}
\usepackage{algorithm}
\usepackage[noend]{algorithmic}
\usepackage{amsmath,amssymb,cite}
\usepackage{footmisc,marvosym}
\usepackage{wasysym,vmargin,stackrel}

\renewcommand{\thefootnote}{\fnsymbol{footnote}}

\newtheorem{theorem}{Theorem}

\newtheorem{corollary}{Corollary}

\newtheorem{definition}{Definition}

\newtheorem{observation}{Observation}

\newtheorem{lemma}{Lemma}

\newtheorem{proposition}{Proposition}

\newenvironment{proof}[1][Proof]{\noindent\textbf{#1.} }{\ \rule{0.5em}{0.5em}}

\setmarginsrb{3.5cm}{2.25cm}{3.5cm}{3.05cm}{0.3cm}{0.3cm}{-0.3cm}{1.0cm}
\begin{document}
\title{Natural Models for Evolution on Networks} 
\author{George B. Mertzios\thanks{School of Engineering and Computing Sciences, Durham University, UK. 
Email: \texttt{george.mertzios@durham.ac.uk}} 
\and Sotiris Nikoletseas\thanks{Computer Technology Institute and University of Patras, Greece.
Email: \texttt{nikole@cti.gr}} 
\and Christoforos Raptopoulos\thanks{Computer Technology Institute and University of Patras, Greece.
Email: \texttt{raptopox@ceid.upatras.gr}} 
\and Paul G. Spirakis\thanks{Computer Technology Institute and University of Patras, Greece.
Email: \texttt{spirakis@cti.gr}}}
\date{\vspace{-1cm}}
\maketitle

\begin{abstract}
Evolutionary dynamics have been traditionally studied in the context of
homogeneous populations, mainly described by the Moran process~\cite{Moran58}. 
Recently, this approach has been generalized in~\cite{Nowak05} by
arranging individuals on the nodes of a network (in general, directed). In
this setting, the existence of directed arcs enables the simulation of
extreme phenomena, where the fixation probability of a randomly placed mutant 
(i.e.~the probability that the offsprings of the mutant eventually spread over the whole population) 
is arbitrarily small or large. 
On the other hand, undirected networks
(i.e.~undirected graphs) seem to have a smoother behavior, and thus it is 
more challenging to find suppressors/amplifiers of selection, that is,
graphs with smaller/greater fixation probability than the complete graph
(i.e.~the homogeneous population). In this paper we focus on undirected
graphs. We present the first class of undirected graphs which act as
suppressors of selection, by achieving a fixation probability that is at most one
half of that of the complete graph, as the number of vertices increases.
Moreover, we provide some generic upper and lower bounds for the fixation
probability of general undirected graphs. As our main contribution, we
introduce the natural alternative of the model proposed in~\cite{Nowak05}. 
In our new evolutionary model, all individuals interact
\emph{simultaneously} and the result is a compromise between aggressive and
non-aggressive individuals. That is, the behavior of the individuals in our
new model and in the model of~\cite{Nowak05} can be interpreted as an 
\emph{``aggregation''} vs.~an \emph{``all-or-nothing''} strategy, respectively. 
We prove that our new model of mutual influences admits a \emph{potential
function}, which guarantees the convergence of the system for any graph
topology and any initial fitness vector of the individuals. Furthermore, we
prove fast convergence to the stable state for the case of the complete graph, 
as well as we provide almost tight bounds on the limit fitness of the individuals. 
Apart from being important on its own, this new evolutionary
model appears to be useful also in the abstract modeling of 
control mechanisms over invading populations in networks. 
We demonstrate this by introducing and analyzing two alternative control approaches, 
for which we bound the time needed to stabilize to the ``healthy'' state of the system.\newline


\end{abstract}



\section{Introduction\label{sec:intro}}

Evolutionary dynamics have been well studied (see~\cite{Traulsen09,Taylor06,Taylor04,Ohtsuki06,Karlin75,Broom09,Antal06}), 
mainly in the context of homogeneous populations, described by the Moran process~\cite{Moran58,NorrisMarkov}. 
In addition, population dynamics have been extensively studied also from the perspective of the 
strategic interaction in evolutionary game theory, cf.~for instance~\cite{Gintis00,Sandholm11,Evolution-Population98,Kandori93,Imhof05}. 
One of the main targets of evolutionary game theory is evolutionary dynamics (see~\cite{Weibull95,Evolution-Population98}). 
Such dynamics usually examine the propagation of intruders with a given \emph{fitness} to a population, 
whose initial members (resident individuals) have a different fitness. 
In fact, ``evolutionary stability'' is the case where no dissident behaviour can invade and dominate the population. 
The evolutionary models and the dynamics we consider here belong to this framework. 
In addition, however, we consider structured populations (i.e.~in the form of an undirected graph) 
and we study how the underlying graph structure affects the evolutionary dynamics. 
We study in this paper two kinds of evolutionary dynamics. Namely, the ``all or nothing'' case 
(where either the intruder overtakes the whole graph or die out) 
and the ``aggregation'' case (more similar in spirit to classical evolutionary game theory, 
where the intruder's fitness aggregates with the population fitness and generates eventually a homogeneous crowd with a new fitness).

In a recent article, Lieberman, Hauert, and Nowak proposed a generalization of the Moran process 
by arranging individuals on a connected network (i.e.~graph)~\cite{Nowak05} (see also~\cite{nowak06-book}). 
In this model, vertices correspond to individuals of the population and 
weighted edges represent the reproductive rates between the adjacent vertices. 
That is, the population structure is translated into
a network (i.e.~graph) structure. Furthermore, individuals (i.e.~vertices)
are partitioned into two types: \emph{aggressive} and \emph{non-aggressive}.
The degree of (relative) aggressiveness of an individual is measured by its 
\emph{relative fitness}; in particular, non-aggressive and aggressive
individuals are assumed to have relative fitness $1$ and $r\geq 1$, respectively. 
This modeling approach initiates an ambitious direction of interdisciplinary research, which combines 
classical aspects of computer science (such as combinatorial structures and complex network topologies), 
probabilistic calculus (discrete Markov chains), and fundamental~aspects of evolutionary game theory (such as evolutionary dynamics).

In the model of~\cite{Nowak05}, one \emph{mutant} (or \emph{invader}) with relative fitness $r\geq 1$ 
is introduced into a given population of \emph{resident} individuals, each of whom having relative fitness $1$. 
For simplicity, a vertex of the graph that is occupied by a mutant will be referred to as \emph{black},
while the rest of the vertices will be referred to as \emph{white}. At each
time step, an individual is chosen for reproduction with a probability
proportional to its fitness, while its offspring replaces a randomly chosen 
neighboring individual in the population. Once $u$ has been selected for 
reproduction, the probability that vertex $u$ places its offspring into 
position $v$ is given by the weight $w_{uv}$ of the directed arc $\langle
uv\rangle $. This process stops when either all vertices of the graph become
black (resulting to a \emph{fixation} of the graph) or they all become white
(resulting to \emph{extinction} of the mutants). 
Several similar models have been previously studied, describing for instance 
influence propagation in social networks (such as the decreasing cascade model~\cite{kempe05,Mossel07}), 
dynamic monopolies~\cite{Berger01}, 
particle interactions (such as the voter model, the antivoter model, and the exclusion process~\cite{Liggett85,Aldous-online-book,Durrett88}),~etc. 
However, the dynamics emerging from these models do not consider different fitnesses for the individuals.

The \emph{fixation probability} $f_{G}$ of a graph $G=(V,E)$ is the
probability that eventually fixation occurs, i.e.~the probability that an
initially introduced mutant, placed uniformly at random on a vertex of~$G$,
eventually spreads over the whole population $V$, replacing all resident
individuals. One of the main characteristics in this model is that at every
iteration of the process, a ``battle'' takes place between aggressive and
non-aggressive individuals, while the process stabilizes only when one of
the two teams of individuals takes over the whole population. This kind of
behavior of the individuals can be interpreted as an \emph{all-or-nothing}
strategy.

Lieberman et al.~\cite{Nowak05} proved that the fixation probability for
every symmetric directed graph (i.e.~when $w_{{uv}}=w_{{vu}}$ for every~$u,v$) 
is equal to that of the complete graph (i.e.~the homogeneous population of
the Moran process), which tends to $1-\frac{1}{r}$ as the size $n$ of the
population grows. Moreover, exploiting vertices with zero in-degree or zero
out-degree (``upstream'' and ``downstream'' populations, respectively), they
provided several examples of \emph{directed} graphs with arbitrarily small
and arbitrarily large fixation probability~\cite{Nowak05}. Furthermore, the
existence of directions on the arcs leads to examples where neither fixation
nor extinction is possible (e.g.~a graph with two sources).

In contrast, general \emph{undirected} graphs (i.e.~when $\langle uv\rangle
\in E$ if and only if $\langle vu\rangle \in E$ for every~$u,v$) appear to
have a smoother behavior, as the above process eventually reaches fixation
or extinction with probability~$1$. Furthermore, the coexistence of both
directions at every edge in an undirected graph seems to make it more
difficult to find \emph{suppressors} or \emph{amplifiers} of selection
(i.e.~graphs with smaller or greater fixation probability than the complete
graph, respectively), or even to derive non-trivial upper and lower bound
for the fixation probability on general undirected graphs. This is the main
reason why only little progress has been done so far in this direction and
why most of the recent work focuses mainly on the exact or numerical
computation of the fixation probability for very special cases of undirected
graphs, e.g.~the star and the path~\cite{Broom08,Broom-two-results,Broom10}.

\vspace{0.2cm} \noindent\textbf{Our contribution.} In this paper we overcome
this difficulty for undirected graphs and we provide the first class of
undirected graphs that act as suppressors of selection in the model of~\cite%
{Nowak05}, as the number of vertices increases. This is a very simple class
of graphs (called \emph{clique-wheels}), where each member $G_{n}$ has a clique of size $n\geq 3$ and an
induced cycle of the same size $n$ with a perfect matching between them. We
prove that, when the mutant is introduced to a clique vertex of $G_{n}$, then the
probability of fixation tends to zero as $n$ grows. Furthermore, we prove
that, when the mutant is introduced to a cycle vertex of $G_{n}$, then the probability
of fixation is at most $1-\frac{1}{r}$ as $n$ grows (i.e.~to the same value
with the homogeneous population of the Moran process). Therefore, since the
clique and the cycle have the same number $n$ of vertices in $G_{n}$, the
fixation probability $f_{G_{n}}$ of $G_{n}$ is at most $\frac{1}{2}(1-\frac{1}{r})$ as $n$ increases,
i.e.~$G_{n}$ is a suppressor of selection. Furthermore, we provide for the model
of~\cite{Nowak05} the first non-trivial upper and lower bounds for the
fixation probability in general undirected graphs. In particular, we first
provide a generic upper bound depending on the degrees of some local
neighborhood. Second, we present another upper and lower bound, depending on
the ratio between the minimum and the maximum degree of the vertices.

As our main contribution, we introduce in this paper the natural alternative 
of the \emph{all-or-nothing} approach of~\cite{Nowak05}, which can be 
interpreted as an \emph{aggregation} strategy. 
In this aggregation model, all individuals interact \emph{simultaneously} 
and the result is a compromise between the aggressive and non-aggressive individuals. 
Both these two alternative models for evolutionary dynamics 
coexist in several domains of interaction between individuals, 
e.g.~in society (dictatorship vs.~democracy, war vs.~negotiation) 
and biology (natural~selection~vs.~mutation~of species). 
In particular, another motivation for our models comes from biological networks, 
in which the interacting individuals (vertices) correspond to cells of an organ 
and advantageous mutants correspond to viral cells or cancer.
Regarding the proposed model of mutual influences, we first prove that it 
admits a \emph{potential} function. This potential function guarantees that
for any graph topology and any initial fitness vector, the system converges
to a stable state, where all individuals have the same fitness. 
%
Furthermore, we analyze the telescopic behavior of this model for the complete graph. 
In particular, we prove fast convergence to the stable state, as well as we provide 
almost tight bounds on the \emph{limit fitness} of the individuals.

Apart from being important on its own, this new evolutionary model enables 
also the abstract modeling of new control mechanisms over invading populations in networks. 
We demonstrate this by introducing and analyzing the behavior of two alternative control approaches. 
In both scenarios we periodically modify the fitness of a small fraction of individuals in the current population, 
which is arranged on a complete graph with $n$ vertices. 
%
In the first scenario, we proceed in phases. Namely, after each modification, 
we let the system stabilize before we perform the next modification. 
In the second scenario, we modify the fitness of a small fraction of individuals at each step. 
In both alternatives, we stop performing these modifications of the population whenever the fitness of every individual 
becomes sufficiently close to $1$ (which is considered to be the ``healthy'' state of the system). 
For the first scenario, we prove that the number of \emph{phases} needed for the system to stabilize in the healthy state 
is logarithmic in $r-1$ and independent of $n$. 
For the second scenario, we prove that the number of \emph{iterations} needed for the system to stabilize in the healthy state 
is linear in $n$ and proportional to $r \ln(r-1)$.

\vspace{0.2cm} \noindent \textbf{Notation.} In an undirected graph $G=(V,E)$, 
the edge between vertices $u\in V$ and $v\in V$ is denoted by $uv\in E$,
and in this case $u$ and $v$ are said to be \emph{adjacent} in $G$. If the
graph $G$ is directed, we denote by ${\langle uv\rangle }$ the arc from $u$
to $v$. For every vertex $u\in V$ in an undirected graph $G=(V,E)$, we
denote by $N(u)=\{v\in V\ |\ uv\in E\}$ the set of neighbors of $u$ in $G$
and by $\deg (u)=|N(u)|$. Furthermore, for any $k\geq 1$, we denote for
simplicity $[k]=\{1,2,\ldots ,k\}$.

\vspace{0.2cm} \noindent \textbf{Organization of the paper.} We discuss in
Section~\ref{sec:all-or-nothing-vs-aggregation} the two alternative models
for evolutionary dynamics on graphs. In particular, we formally present in
Section~\ref{Nowak-model-subsec} the model of~\cite{Nowak05} and then we
introduce in Section~\ref{mutual-model-subsec} our new model of mutual
influences. In Section~\ref{sequential-sec} we first provide generic upper
and lower bounds of the fixation probability in the model of~\cite{Nowak05}
for arbitrary undirected graphs. Then we present in Section~\ref%
{suppressor-subsec} the first class of undirected graphs which act as
suppressors of selection in the model of~\cite{Nowak05}, as the number of
vertices increases. In Section~\ref{parallel-model-sec} we analyze our new
evolutionary model of mutual influences. In particular, we first prove in
Section~\ref{potential-subsec} the convergence of the model by using a
potential function, and then we analyze in Section~\ref{clique-mutual-subsec} 
the telescopic behavior of this model for the case of a complete graph. 
In Section~\ref{antibiotics-sec} we demonstrate the use of our new model
in analyzing the behavior of two alternative invasion control mechanisms. Finally, we
discuss the presented results and further research in Section~\ref{conclusions}.

\section{All-or-nothing vs. aggregation\label%
{sec:all-or-nothing-vs-aggregation}}

In this section we formally define the model of~\cite{Nowak05} for
undirected graphs and we introduce our new model of mutual influences.
Similarly to~\cite{Nowak05}, we assume for every edge $uv$ of an undirected graph 
that $w_{uv}=\frac{1}{\deg u}$ and $w_{vu}=\frac{1}{\deg v}$,
i.e.~once a vertex $u$ has been chosen for reproduction, it chooses one of
its neighbors uniformly at random.

\subsection{The model of Lieberman, Hauert, and Nowak (an all-or-nothing
approach)\label{Nowak-model-subsec}}

Let $G=(V,E)$ be a connected undirected graph with $n$ vertices. Then, the
stochastic process defined in~\cite{Nowak05} can be described by a Markov
chain with state space $\mathcal{S}=2^{V}$ (i.e.~the set of all subsets of~$V$) and
transition probability matrix $P$, where for any two states $S_{1},S_{2}\subseteq V$,

\begin{equation}
P_{S_{1},S_{2}}=\left\{ 
\begin{array}{ll}
\vspace{0.1cm}\frac{1}{|S_{1}|r+n-|S_{1}|}\cdot \sum\limits_{u\in N(v)\cap
S_{1}}\frac{r}{\deg (u)}\text{,} & \quad \text{if $S_{2} = S_{1} \cup \{v\}$ 
and $v\notin S_{1}$} \\ 
\vspace{0.1cm}\frac{1}{|S_{1}|r+n-|S_{1}|}\cdot \sum\limits_{u\in
N(v)\setminus S_{2}}\frac{1}{\deg (u)}\text{,} & \quad \text{if $S_{1} = S_{2} \cup \{v\}$ 
and $v\notin S_{2}$} \\ 
\frac{1}{|S_{1}|r+n-|S_{1}|}\cdot \left( \sum\limits_{u\in S_{1}}\frac{%
r\cdot |N(u)\cap S_{1}|}{\deg (u)}+\sum\limits_{u\in V\setminus S_{1}}\frac{%
|N(u)\cap (V\setminus S_{1})|}{\deg (u)}\right) \text{,} & \quad \text{if }%
S_{2}=S_{1} \\ 
0\text{,} & \quad \text{otherwise}%
\end{array}%
\right.  \label{P-transition-Nowak}
\end{equation}

Notice that in the above Markov chain there are two absorbing states, namely 
$\emptyset $ and $V$, which describe the cases where the vertices of $G$ are
all white or all black, respectively. Since $G$ is connected, the above
Markov chain will eventually reach one of these two absorbing states with
probability~$1$. If we denote by $h_{v}$ the probability of absorption at
state $V$, given that we start with a single mutant placed initially on
vertex $v$, then by definition $f_{G}=\frac{\sum_{v}h_{v}}{n}$. Generalizing
this notation, let $h_{S}$ be the probability of absorption at $V$ given
that we start at state $S\subseteq V$, and let $h=[h_{S}]_{S\subseteq V}$.
Then, it follows that vector $h$ is the unique solution of the linear system 
$h=P\cdot h$ with boundary conditions $h_{\emptyset }=0$ and $h_{V}=1$.

However, observe that the state space $\mathcal{S}=2^{V}$ of this Markov
chain has size $2^{n}$, i.e.~the matrix $P=[P_{S_{1},S_{2}}]$ in (\ref%
{P-transition-Nowak}) has dimension $2^{n}\times 2^{n}$. This indicates that
the problem of computing the fixation probability $f_{G}$ of a given graph $%
G $ is hard, as also mentioned in~\cite{Nowak05}. This is the main reason
why, to the best of our knowledge, all known results so far regarding the\
computation of the fixation probability of undirected graphs are restricted
to regular graphs, stars, and paths~\cite%
{nowak06-book,Nowak05,Broom08,Broom-two-results,Broom10}. In particular, for
the case of regular graphs, the above Markov chain is equivalent to a
birth-death process with $n-1$ transient (non-absorbing) states, where the
forward bias at every state (i.e.~the ratio of the forward probability over
the backward probability) is equal to $r$. In this case, the fixation
probability is equal to 
\begin{equation}
\rho =\frac{1}{1+\sum_{i=1}^{n-1}\frac{1}{r^{i}}}=\frac{1-\frac{1}{r}}{1-%
\frac{1}{r^{n}}}  \label{rho-general}
\end{equation}%
cf.~\cite{nowak06-book}, chapter 8. It is worth mentioning that, even for
the case of paths, there is no known exact or approximate formula for the
fixation probability~\cite{Broom08}.

\subsection{An evolutionary model of mutual influences (an aggregation 
approach)\label{mutual-model-subsec}}

The evolutionary model of~\cite{Nowak05} constitutes a sequential process,
in every step of which only two individuals interact and the process
eventually reaches one of two extreme states. However, in many evolutionary
processes, all individuals may interact simultaneously at each time step,
while some individuals have greater influence to the rest of the population
than others. This observation leads naturally to the following model for
evolution on graphs, which can be thought as a smooth version of the
model presented in~\cite{Nowak05}.

Consider a population of size~$n$ and a portion $\alpha \in \lbrack 0,1]$ of
newly introduced mutants with relative fitness~$r$. The topology of the
population is given in general by a directed graph~$G=(V,E)$ with~$|V|=n$
vertices, where the directed arcs of~$E$ describe the allowed interactions
between the individuals. At each time step, \emph{every} individual~${u\in V}$
of the population influences every individual~${v\in V}$, for which~$\langle
uv\rangle \in E$, while the degree of this influence is proportional to the
fitness of~$u$ and to the weight~$w_{uv}$ of the arc~$\langle uv\rangle $.
Note that we can assume without loss of generality that the weights $w_{uv}$
on the arcs are normalized, i.e.~for every fixed vertex $u\in V$ it holds $%
\sum_{\langle uv\rangle \in E}w_{uv}=1$ . Although this model can be defined
in general for directed graphs with arbitrary arc weights $w_{uv}$, we will
focus in the following to the case where $G$ is an undirected graph (i.e.~$%
\langle u_{i}u_{j}\rangle \in E$ if and only if~${\langle u_{j}u_{i}\rangle
\in E}$, for every~$i,j$) and $w_{uv}=\frac{1}{\deg (u)}$ for all edges $%
uv\in E$.

Formally, let $V=\{u_{1},u_{2},\ldots ,u_{n}\}$ be the set of vertices and $%
r_{u_{i}}(k)$ be the fitness of the vertex $u_{i}\in V$ at iteration $k\geq
0 $. Let $\Sigma (k)$ denote the sum of the fitnesses of all vertices at
iteration $k$, i.e.~$\Sigma (k)=\sum_{i=1}^{n}r_{u_{i}}(k)$. Then the vector 
$r(k+1)$ with the fitnesses $r_{u_{i}}(k+1)$ of the vertices $u_{i}\in V$ at
the next iteration $k+1$ is given by
\begin{equation}
\lbrack r_{u_{1}}(k+1),r_{u_{2}}(k+1),\ldots ,r_{u_{n}}(k+1)]^{T}=P\cdot
\lbrack r_{u_{1}}(k),r_{u_{2}}(k),\ldots ,r_{u_{n}}(k)]^{T}
\label{recursion-P}
\end{equation}
i.e.%
\begin{equation}
r(k+1)=P\cdot r(k)  \label{recursion-P-compact}
\end{equation}%
In the latter equation,  the elements of the square matrix $P=[P_{ij}]_{i,j=1}^{n}$ depend on the iteration $k$ and they are given as follows:

\begin{equation}
P_{ij}=\left\{ 
\begin{array}{ll}
\vspace{0.1cm}\frac{r_{u_{j}}(k)}{\deg (u_{j})\Sigma (k)}\text{,} & \quad 
\text{if }i\neq j\text{ and }u_{i}u_{j}\in E \\ 
\vspace{0.1cm}0\text{,} & \quad \text{if }i\neq j\text{ and }%
u_{i}u_{j}\notin E \\ 
1-\sum_{j\neq i}P_{ij}\text{,} & \quad \text{if }i=j%
\end{array}%
\right.  \label{P-exact}
\end{equation}%
Note by (\ref{recursion-P-compact}) and (\ref{P-exact}) that after the first
iteration, the fitness of every individual in our new evolutionary model of
mutual influences equals the expected fitness of this individual in the
model of~\cite{Nowak05} (cf.~Section~\ref{Nowak-model-subsec}). However, this correlation of the two models is not
maintained in the next iterations and the two models behave differently as
the processes evolve.

In particular, in the case where $G$ is the complete graph, i.e.~$\deg
(u_{i})=n-1$ for every vertex~$u_{i}$, the matrix $P$ becomes

\begin{equation}
P=\left[ 
\begin{array}{cccc}
\vspace{0.2cm}1-\frac{r_{u_{2}}(k)+\ldots +r_{u_{n}}(k)}{(n-1)\Sigma (k)} & 
\frac{r_{u_{2}}(k)}{(n-1)\Sigma (k)} & \quad \cdots \quad  & \frac{%
r_{u_{n}}(k)}{(n-1)\Sigma (k)} \\ 
\vspace{0.2cm}\frac{r_{u_{1}}(k)}{(n-1)\Sigma (k)} & 1-\frac{%
r_{u_{1}}(k)+r_{u_{3}}(k)+\ldots +r_{u_{n}}(k)}{(n-1)\Sigma (k)} & \quad
\cdots \quad  & \frac{r_{u_{n}}(k)}{(n-1)\Sigma (k)} \\ 
\vspace{0.2cm}\cdots  & \cdots  & \quad \cdots \quad  & \cdots  \\ 
\frac{r_{u_{1}}(k)}{(n-1)\Sigma (k)} & \frac{r_{u_{2}}(k)}{(n-1)\Sigma (k)}
& \quad \cdots \quad  & 1-\frac{r_{u_{1}}(k)+\ldots +r_{u_{n-1}}(k)}{%
(n-1)\Sigma (k)}%
\end{array}%
\right]   \label{P-clique}
\end{equation}%
The system given by (\ref{recursion-P-compact}) and (\ref{P-exact}) can be
defined for every initial fitness vector $r(0)$. However, in the case where there is initially a portion $\alpha \in \lbrack 0,1]$ of newly
introduced mutants with relative fitness $r$, the initial condition $r(0)$
of the system in (\ref{recursion-P}) is a vector with $\alpha n$ entries
equal to $r$ and with $(1-\alpha )n$ entries equal to $1$.

\begin{observation}
\label{quadratic-obs} Note that the recursive equation~(\ref%
{recursion-P-compact}) is a \emph{non-linear} equation on the fitness values~%
$r_{u_{j}}(k)$ of the vertices at iteration $k$.
\end{observation}

Since by (\ref{P-exact}) the sum of every row of the matrix $P$ equals to one, 
the fitness $r_{u_{i}}(k)$ of vertex~$u_{i}$ after the ${(k+1)}$-th
iteration of the process is a convex combination of the fitnesses of the
neighbors of $u_{i}$ after the $k$-th iteration. Therefore, in particular,
the fitness of every vertex $u_{i}$ at every iteration~${k\geq 0}$ lies
between the smallest and the greatest initial fitness of the vertices, as
the next observation states.

\begin{observation}
\label{convex-bound-parallel-obs} Let $r_{\min }$ and $r_{\max }$ be the
smallest and the greatest initial fitness in $r(0)$, respectively. 
Then $r_{\min }\leq r_{u_{i}}(k)\leq r_{\max }$ for every $u_{i}\in V$ and every $k\geq 0$.
\end{observation}

\paragraph{Degree of influence.}

Suppose that initially $\alpha n$ mutants (for some ${\alpha \in \lbrack 0,1]%
}$) with relative fitness~${r\geq 1}$ are introduced in graph $G$ on a
subset ${S\subseteq V}$ of its vertices. Then, as we prove in Theorem~\ref%
{undir-conv-thm}, after a certain number of iterations the fitness vector $%
r(k)$ converges to a vector $[r_{0}^{S},r_{0}^{S},\ldots ,r_{0}^{S}]^{T}$,
for some value $r_{0}^{S}$. This \emph{limit fitness} $r_{0}^{S}$ depends in
general on the initial relative fitness $r$ of the mutants, on their initial
number $\alpha n$, as well as on their initial position on the vertices of $%
S\subseteq V$. The relative fitness $r$ of the initially introduced mutants
can be thought as having the ``black'' color, while the initial fitness of
all the other vertices can be thought as having the ``white'' color. Then,
the limit fitness $r_{0}^{S}$ can be thought as the ``degree of gray color''
that all the vertices obtain after sufficiently many iterations, given that
the mutants are initially placed at the vertices of~$S$. In the case where
the $\alpha n$ mutants are initially placed with \emph{uniform} probability
to the vertices of~$G$, we can define the \emph{limit fitness $r_{0}$ of~$G$} as

\begin{equation}
r_{0}=\frac{\sum\limits_{S\subseteq V,\ |S|=\alpha n}r_{0}^{S}}{{\binom{n}{\alpha n}}}
\end{equation}%
For a given initial value of $r$, the bigger is $r_{0}$ the stronger is the
effect of natural selection in $G$.

Since $r_{0}^{S}$ is a convex combination of $r$ and $1$, there exists a
value $f_{G,S}(r)\in \lbrack 0,1]$, such that ${r_{0}^{S}=f_{G,S}(r)\cdot r+(1-f_{G,S}(r))\cdot 1}$. 
Then, the value $f_{G,S}(r)$ is the \emph{degree of influence} of the graph~$G$, 
given that the mutants are initially placed at the vertices of $S$. In
the case where the mutants are initially placed with uniform probability at
the vertices of $G$, we can define the degree of influence of $G$ as\vspace{-0.3cm}

\begin{equation}
f_{G}(r)=\frac{\sum\limits_{S\subseteq V,\ |S|=\alpha n}f_{G,S}(r)}{{\binom{n%
}{\alpha n}}}
\end{equation}

\paragraph{Number of iterations to stability.}

For some graphs $G$, the fitness vector $r(k)$ reaches \emph{exactly} the 
\emph{limit fitness vector} $[r_{0},r_{0},\ldots ,r_{0}]^{T}$ (for instance,
the complete graph with two vertices and one mutant not only reaches this
limit in exactly one iteration, but also the degree of influence is exactly
the fixation probability of this simple graph). However, for other graphs $G$
the fitness vector $r(k)$ converges to $[r_{0},r_{0},\ldots ,r_{0}]^{T}$
(cf.~Theorem~\ref{undir-conv-thm} below), but it never becomes equal to it.
In the first case, one can compute (exactly or approximately) the number of
iterations needed to reach the limit fitness vector. 
In the second case, given an arbitrary 
$\varepsilon >0$, one can compute the number of
iterations needed to come $\varepsilon $-close to the limit fitness vector.

\section{Analysis of the all-or-nothing model\label{sequential-sec}}

In this section we present analytic results on the evolutionary model of 
\cite{Nowak05}, which is based on the sequential interaction among the
individuals. In particular, we first present non-trivial upper and lower
bounds for the fixation probability, depending on the degrees of vertices. 
Then we present the first class of
undirected graphs that act as suppressors of selection in the model of~\cite%
{Nowak05}, as the number of vertices increases.

Recall by the preamble of Section~\ref{mutual-model-subsec} that, similarly
to~\cite{Nowak05}, we assumed that $w_{uv}=\frac{1}{\deg u}$ and $w_{vu}=%
\frac{1}{\deg v}$ for every edge $uv$ of an undirected graph $G=(V,E)$. It
is easy to see that this formulation is equivalent to assigning to every
edge $e=uv\in E$ the weight $w_{e}=w_{uv}=w_{vu}=1$, since also in this
case, once a vertex $u$ has been chosen for reproduction, it chooses one of
its neighbors uniformly at random. A natural generalization of this weight
assignment is to consider~$G$ as a complete graph, where every edge $e$ in
the clique is assigned a non-negative weight $w_{e}\geq 0$, and $w_{e}$ is
not necessarily an integer. Note that, whenever $w_{e}=0$, it is as if the
edge $e$ is not present in $G$. Then, once a vertex $u$ has been chosen for
reproduction, $u$ chooses any other vertex~$v$ with probability $\frac{w_{uv}%
}{\sum_{x\neq u}w_{ux}}$.

Note that, if we do not impose any additional constraint on the weights, we
can simulate multigraphs by just setting the weight of an edge to be equal to the
multiplicity of this edge. Furthermore, we can construct graphs with
arbitrary small fixation probability. For instance, consider an undirected
star with $n$ leaves, where one of the edges has weight an arbitrary small $%
\varepsilon >0$ and all the other edges have weight $1$. Then, the leaf that
is incident to the edge with weight $\varepsilon $ acts as a source in the
graph as $\varepsilon \rightarrow 0$. Thus, the only chance to reach
fixation is when we initially place the mutant at the source, i.e.~the
fixation probability of this graph tends to $\frac{1}{n+1}$ as $\varepsilon
\rightarrow 0$. Therefore, it seems that the difficulty to construct strong
suppressors lies in the fact that unweighted undirected graphs can not
simulate sources. For this reason, we consider in the remainder of this
paper only unweighted undirected graphs.

\subsection{A generic upper bound approach\label{generic-upper-Nowak-subsec}}

In the next theorem we provide a generic upper bound of the fixation
probability of undirected graphs, depending on the degrees of the vertices
in some local neighborhood.

\begin{theorem}
\label{theoremwithQ}Let $G=(V,E)$ be an undirected graph. For any $uv\in E$,
let $Q_{u}=\sum_{x\in N(u)}\frac{1}{\deg {x}}$ and $Q_{uv}=\sum_{x\in
N(u)\backslash \{v\}}\frac{1}{\deg {x}}+\sum_{x\in N(v)\backslash \{u\}}%
\frac{1}{\deg {x}}$. Then%
\begin{equation}
f_{G}\leq \max_{uv\in E}\left\{ \frac{r^{2}}{r^{2}+rQ_{u}+\frac{Q_{u}Q_{uv}}{%
2}}\right\}   \label{generic-upper}
\end{equation}
\end{theorem}

\begin{proof}
For the proof we construct a simple Markov chain $\tilde{\mathcal{M}}$, in
which the probability of reaching one of its absorbing states is at least
the probability of fixation in the original Markov chain. Then, in order to
provide an upper bound of the fixation probability in the original Markov
chain, we provide an upper bound on the probability of reaching one of the
absorbing states in $\tilde{\mathcal{M}}$.

Let $u$ be a vertex that maximizes the probability of fixation, that is, $u
\in \arg \max_{u\in V}h_{u}$. Furthermore, assume that we end the process in
favor of the black vertices when the corresponding Markov chain describing
the model of~\cite{Nowak05} reaches three black vertices. To favor fixation
even more, since $u$ maximizes $h_{u}$, we assume that, whenever we reach
two black vertices and a backward step happens (i.e.,~a step that reduces
the number of black vertices), then we backtrack to state $u$ (even if
vertex $u$ was the one that became white). Finally, given that we start at
vertex $u$ and we increase the number of black vertices by one, we assume
that we make black the neighbor $v$ of $u$ that maximizes the forward bias
of the state $\{u,v\} $. Imposing these constraints (and eliminating self
loops), we get a Markov chain $\tilde{\mathcal{M}}$, shown in Figure~\ref%
{upper3bound}, that dominates the original Markov chain. That is, the
probability that $\tilde{\mathcal{M}}$ reaches the state of three black
vertices, given that we start at $u$, is an upper bound of the fixation
probability of $G$.

\begin{figure}[htb]
\centering
\includegraphics[scale=0.695]{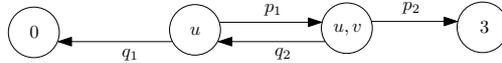}
\caption{The Markov chain $\tilde{\mathcal{M}}$.}
\label{upper3bound}
\end{figure}
For the Markov chain $\tilde{\mathcal{M}}$, we have that

\begin{equation*}
q_{1}=\frac{\sum_{x\in N(u)}\frac{1}{\deg {x}}}{r+\sum_{x\in N(u)}\frac{1}{%
\deg {x}}}\overset{def}{=}\frac{Q_{u}}{r+Q_{u}}=1-p_{1}
\end{equation*}%
where $N(u)$ is the set of neighbors of $u$. Also,

\begin{eqnarray}
q_{2} &=&\frac{\sum_{x\in N(u)\backslash \{v\}}\frac{1}{\deg {x}}+\sum_{x\in
N(v)\backslash \{u\}}\frac{1}{\deg {x}}}{r(1-\frac{1}{\deg u})+r(1-\frac{1}{%
\deg v})+\sum_{x\in N(u)\backslash \{v\}}\frac{1}{\deg {x}}+\sum_{x\in
N(v)\backslash \{u\}}\frac{1}{\deg {x}}}  \notag \\
&\overset{def}{=}&\frac{Q_{uv}}{r(2-\frac{1}{\deg u}-\frac{1}{\deg v})+Q_{uv}%
}=1-p_{2}  \notag
\end{eqnarray}%
Let now $\tilde{h}_{u}$ (resp.~$\tilde{h}_{uv}$) denote the probability of
reaching 3 blacks, starting from $u$ (resp.~starting from the state $\{u,v\}$%
) in $\tilde{\mathcal{M}}$. We have that 
\begin{eqnarray*}
\tilde{h}_{u} &=&p_{1}\tilde{h}_{uv}=p_{1}(p_{2}+q_{2}\tilde{h}%
_{u})\Leftrightarrow  \\
\tilde{h}_{u} &=&\frac{p_{1}p_{2}}{1-p_{1}q_{2}}=\frac{r^{2}}{r^{2}+rQ_{u}+%
\frac{Q_{u}Q_{uv}}{2-\frac{1}{\deg u}-\frac{1}{\deg v}}}\leq \frac{r^{2}}{%
r^{2}+rQ_{u}+\frac{Q_{u}Q_{uv}}{2}}
\end{eqnarray*}%
This completes the proof of the theorem.
\end{proof}

\medskip

Consider for instance a bipartite graph $G=(U,V,E)$, where $\deg u=d_{1}$
for every vertex $u\in U$ and $\deg v=d_{2}$ for every vertex $v\in V$. Then
any edge of $E$ has one vertex in $U$ and one vertex in~$V$. Using the above
notation, consider now an arbitrary edge $uv\in E$, where $u\in U$ and $v\in
V$. Then $Q_{u}=\frac{d_{1}}{d_{2}}$ and $Q_{uv}=\frac{d_{1}-1}{d_{2}}+\frac{%
d_{2}-1}{d_{1}}$. The right side of (\ref{generic-upper})
is maximized when $d_{1}<d_{2}$, and thus in this case Theorem~\ref%
{theoremwithQ} implies that $f_{G}\leq \frac{r^{2}}{r^{2}+r\frac{d_{1}}{d_{2}%
}+\frac{d_{1}}{d_{2}}\left( \frac{d_{1}-1}{d_{2}}+\frac{d_{2}-1}{d_{1}}%
\right) }$. In particular, for the star graph with~${n+1}$ vertices, we have $%
d_{1}=1$ and $d_{2}=n$. But, as shown in~\cite{Nowak05}, the fixation
probability of the star is asymptotically equal to $1-\frac{1}{r^{2}}$,
whereas the above bound gives $f_{star}\leq \frac{r^{2}}{r^{2}+r\frac{1}{n}+%
\frac{n-1}{n}}=1-\frac{1}{r^{2}+1+o(1)}$, which is quite tight.

\subsection{Upper and lower bounds depending on degrees\label%
{delta-Delta-upper-lower-subsec}}

In the following theorem we provide upper and lower bounds of the fixation
probability of undirected graphs, depending on the minimum and the maximum
degree of the vertices.

\begin{theorem}
\label{worstcasetheorem}Let $G=(V,E)$ be an undirected graph, where $\frac{%
\deg (v)}{\deg (u)}\leq \lambda $ whenever $uv\in E$. Then, the fixation
probability $f_{G}$ of $G$, when the fitness of the mutant is $r$, is upper
(resp.~lower) bounded by the fixation probability of the clique for mutant
fitness $r_{1}=r\lambda $ (resp.~for mutant fitness $r_{2}=\frac{r}{\lambda }
$). That is,%
\begin{equation}
\frac{1-\frac{\lambda}{r}}{1-(\frac{\lambda}{r})^{n}}
\leq f_{G}\leq 
\frac{1-\frac{1}{r\lambda }}{1-(\frac{1}{r\lambda})^{n}}
\end{equation}
\end{theorem}

\begin{proof}
For an arbitrary state $S\subseteq V$ of the Markov Chain (that corresponds
to the set of black vertices in that state), let $\rho _{+}(S)$ (resp.~$\rho
_{-}(S)$) denote the probability that the number of black vertices increases
(resp.~decreases). In the case where $G$ is the clique, the forward bias at
state $S$ is equal to $\frac{\rho _{+}(S)}{\rho _{-}(S)}=r$, for every state 
$S$~\cite{nowak06-book,Nowak05}. Let $C_{S}=\{uv\in E\ |\ u\in S,\ v\notin
S\}$ be the set of edges with one vertex in $S$ and one vertex in $%
V\setminus S$. Then,%
\begin{equation}
\rho _{+}(S)=\sum_{\{uv\in E\ |\ u\in S,\ v\notin S\}}\frac{r}{n-|S|+r|S|}%
\frac{1}{\deg (u)}  \label{rho-plus-1}
\end{equation}%
and 
\begin{equation}
\rho _{-}(S)=\sum_{\{uv\in E\ |\ u\in S,\ v\notin S\}}\frac{1}{n-|S|+r|S|}%
\frac{1}{\deg (v)}  \label{rho-minus-1}
\end{equation}

Now, since by assumption $\frac{\deg (v)}{\deg (u)}\leq \lambda $ whenever $%
uv\in E$, it follows that%
\begin{equation}
\frac{1}{\lambda }\cdot \sum_{\{uv\in E\ |\ u\in S,\ v\notin S\}}\frac{1}{%
\deg (v)}\leq \sum_{\{uv\in E\ |\ u\in S,\ v\notin S\}}\frac{1}{\deg (u)}%
\leq \lambda \cdot \sum_{\{uv\in E\ |\ u\in S,\ v\notin S\}}\frac{1}{\deg (v)%
}  \label{rho-ineq}
\end{equation}

By (\ref{rho-plus-1}), (\ref{rho-minus-1}), and (\ref{rho-ineq}) we get the
following upper and lower bounds for the forward bias at state $S$.%
\begin{equation}
\frac{r}{\lambda }\leq \frac{\rho _{+}(S)}{\rho _{-}(S)}\leq r\lambda 
\label{rho-3}
\end{equation}%
Notice that the upper and lower bounds of (\ref{rho-3}) for the forward bias
at state $S$ are independent of $S$. Therefore, the process stochastically
dominates a birth-death process with forward bias $\frac{r}{\lambda }$,
while it is stochastically dominated by a birth-death process with forward
bias $r\lambda $ (cf.~equation~(\ref{rho-general})). This completes the
proof of the theorem.
\end{proof}

\subsection{The undirected suppressor\label{suppressor-subsec}}

In this section we provide the first class of undirected graphs (which we
call \emph{clique-wheels}) that act as suppressors of selection as the
number of vertices increases. In particular, we prove that the fixation
probability of the members of this class is at most $\frac{1}{2}(1-\frac{1}{r%
})$, i.e.~the half of the fixation probability of the complete graph, as $%
n\rightarrow \infty $. An example of a clique-wheel graph $G_{n}$ is
depicted in Figure~\ref{axinoswithring}. This graph consists of a clique of
size~$n\geq 3$ and an induced cycle of the same size~$n$ with a perfect
matching between them. We will refer in the following to the vertices of the
inner clique as \emph{clique vertices} and to the vertices of the outer
cycle as \emph{ring vertices}.

%
\begin{figure}[h!tb]
\centering 
\subfigure[]{ \label{axinoswithring}
\includegraphics[scale=0.8]{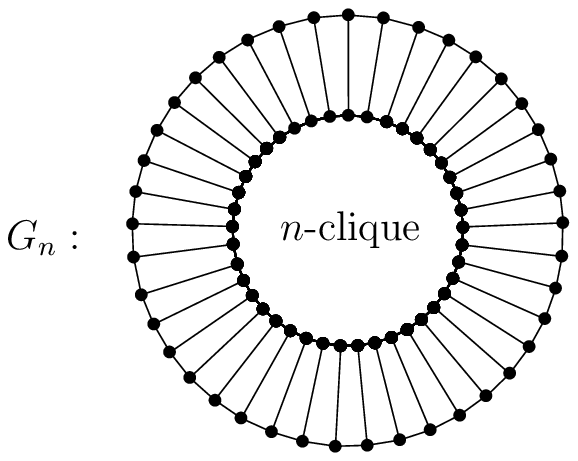}} \hspace{0.2cm} 
\subfigure[]{ \label{state-graph-nose-clique}
\includegraphics[scale=0.85]{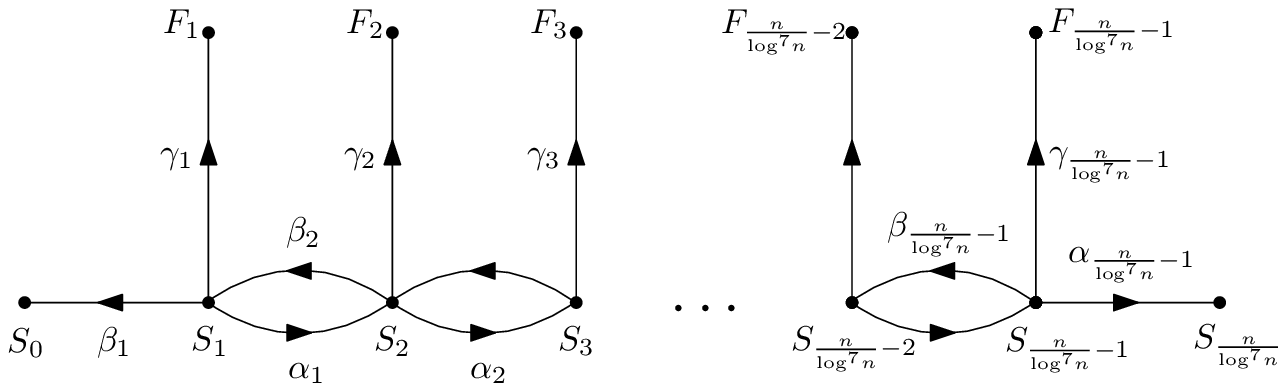}} 
\caption{(a)~The clique-wheel graph $G_{n}$ and (b)~the state graph of a
relaxed Markov chain for computing an upper bound of $h_{1}=h_{clique}$.}
\end{figure}

%

Denote by $h_{clique}$ (resp.~$h_{ring}$) the probability that all the
vertices of $G_{n}$ become black, given that we start with one black clique vertex
(resp.~with one black ring vertex). We first provide in the next lemma an
upper bound on $h_{clique}$. 

\begin{lemma}
\label{lemmaforhclique}For any $r\in \left( 1,\frac{4}{3}\right) $,%
\begin{equation*}
h_{clique}\leq \frac{7}{6n\left( \frac{4}{3r}-1\right) }+o\left( \frac{1}{n}%
\right)
\end{equation*}
\end{lemma}

\begin{proof}
Denote by $S_{k}$ the state, in which exactly $k\geq 0$ clique vertices are
black and all ring vertices are white. Note that $S_{0}$ is the empty state. Denote
by $F_{k}$ the state where at least one ring vertex of $G_{n}$ and exactly $k\geq 0$
clique vertices are black. Furthermore, for every state $S$, denote by $%
h_{S} $ the probability that, starting at the state $S$, we eventually reach
the full state (i.e.~the state where all vertices are black). In the
following we denote for simplicity the probability $h_{S_{k}}$ by $h_{k}$,
for every $k=0,1,\ldots ,n$. Note that $h_{0}=0$ and $h_{1}=h_{clique}$,
since $S_{0}$ is the empty state and $S_{1}$ is the state with only one
black clique vertex. In order to compute an upper bound of $h_{1}$, we can
set the value $h_{\frac{n}{\log ^{7}n}}$ and the values $h_{F_{k}}$ for
every $k\geq 1$ to their trivial upper bound $1$. That is, we assume that
the state $S_{\frac{n}{\log ^{7}{n}}}$, as well as all states $F_{k}$, where 
$k\geq 1$, are absorbing. After performing these relaxations (and eliminating self loops), 
we obtain a Markov chain, whose state graph is illustrated in Figure~\ref{state-graph-nose-clique}. 

For any $k=1,\ldots ,\frac{n}{\log ^{7}{n}}-1$ in this Markov chain,%
\begin{equation}
h_{k}=\alpha _{k}h_{k+1}+\beta _{k}h_{k-1}+\gamma _{k}  \label{eqforhk}
\end{equation}%
where%
\begin{eqnarray}
\alpha _{k} &=&\frac{r\frac{k(n-k)}{n}}{r\frac{k(n-k+1)}{n}+k\left( \frac{1}{%
3}+\frac{n-k}{n}\right) }  \notag \\
\beta _{k} &=&\frac{k\left( \frac{1}{3}+\frac{n-k}{n}\right) }{r\frac{%
k(n-k+1)}{n}+k\left( \frac{1}{3}+\frac{n-k}{n}\right) }
\label{alpha-beta-gamma} \\
\gamma _{k} &=&\frac{r\frac{k}{n}}{r\frac{k(n-k+1)}{n}+k\left( \frac{1}{3}+%
\frac{n-k}{n}\right) }  \notag
\end{eqnarray}%
Notice now by (\ref{alpha-beta-gamma}) that%
\begin{equation}
\frac{\beta _{k}}{\alpha _{k}}=\frac{\frac{4}{3}n-k}{r(n-k)}\geq \frac{4}{3r}%
>1  \label{beta-alpha}
\end{equation}%
since $r\in \left( 1,\frac{4}{3}\right) $ by assumption. Furthermore, since $%
\frac{1}{1-\frac{1}{\log ^{7}n}}\leq \frac{7}{6}$ for sufficiently large $n$%
, it follows that for every $k=1,2,\ldots ,\frac{n}{\log ^{7}{n}}-1$,%
\begin{equation}
\frac{\gamma _{k}}{\alpha _{k}}=\frac{1}{n-k}\leq \frac{7}{6n}
\label{gamma-alpha}
\end{equation}%
Now, since $\alpha _{k}+\beta _{k}+\gamma _{k}=1$, (\ref{eqforhk}) implies
by (\ref{beta-alpha}) and (\ref{gamma-alpha}) that%
\begin{eqnarray*}
h_{k+1}-h_{k} &=&\frac{\beta _{k}}{\alpha _{k}}(h_{k}-h_{k-1})-\frac{\gamma
_{k}}{\alpha _{k}}(1-h_{k}) \\
&\geq &\frac{4}{3r}(h_{k}-h_{k-1})-\frac{7}{6n}
\end{eqnarray*}%
Thus, since $h_{0}=0$ and $h_{k}\geq h_{k-1}$ for all $k=1,\dots ,\frac{n}{%
\log ^{7}{n}}$, it follows that for every $k$,%
\begin{eqnarray*}
h_{k+1}-h_{k} &\geq &\left( \frac{4}{3r}\right) ^{k}(h_{1}-h_{0})-\frac{7}{6n%
}\cdot \sum_{i=0}^{k-1}(\frac{4}{3r})^{i} \\
&=&\left( \frac{4}{3r}\right) ^{k}h_{1}-\frac{7}{6n}\cdot \frac{(\frac{4}{3r}%
)^{k}-1}{\frac{4}{3r}-1}
\end{eqnarray*}%
Consequently, since $h_{\frac{n}{\log ^{7}n}}=1$ in the relaxed Markov
chain, we have that%
\begin{eqnarray*}
1-h_{1} &=&\sum_{k=1}^{\frac{n}{\log ^{7}{n}}-1}(h_{k+1}-h_{k}) \\
&\geq &\sum_{k=1}^{\frac{n}{\log ^{7}{n}}-1}\left[ \left( \frac{4}{3r}%
\right) ^{k}h_{1}-\frac{7}{6n}\cdot \frac{(\frac{4}{3r})^{k}-1}{\frac{4}{3r}%
-1}\right] \Rightarrow \\
h_{1}\sum_{k=0}^{\frac{n}{\log ^{7}{n}}-1}\left( \frac{4}{3r}\right) ^{k}
&\leq &1+\frac{7}{6n\left( \frac{4}{3r}-1\right) }\sum_{k=0}^{\frac{n}{\log
^{7}{n}}-1}\left[ \left( \frac{4}{3r}\right) ^{k}-1\right]
\end{eqnarray*}%
and thus%
\begin{equation*}
h_{1}\leq \frac{7}{6n\left( \frac{4}{3r}-1\right) }+\frac{1}{\sum_{k=0}^{%
\frac{n}{\log ^{7}{n}}-1}\left( \frac{4}{3r}\right) ^{k}}
\end{equation*}%
This completes the proof of the lemma, since $\frac{4}{3r}>1$.
\end{proof}

\medskip

The next corollary follows by the proof of Lemma~\ref{lemmaforhclique}.

\begin{corollary}
\label{interpretation}Starting with one black clique vertex, the probability
that we make at least one ring vertex black, or that we eventually reach $\frac{n}{%
\log ^{7}{n}}$ black clique vertices, is at most $\frac{7}{6n\left( \frac{4}{%
3r}-1\right) }+o\left( \frac{1}{n}\right) $.
\end{corollary}

In the remainder of this section, we will also provide an upper bound on $h_{ring}$, 
thus bounding the fixation probability $f_{G_{n}}$ of $G_{n}$ (cf.~Theorem~\ref{bound-fixation-thm}). 
Consider the Markov chain $\mathcal{M}$ that is
depicted in Figure~\ref{MCM}. Our analysis will use the following auxiliary
lemma which concerns the expected time to absorption of this Markov chain. 
\begin{figure}[tbh]
\centering
\includegraphics[scale=0.695]{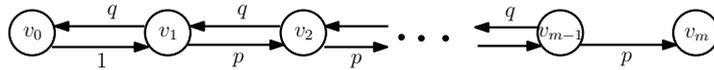}
\caption{The Markov chain $\mathcal{M}$.}
\label{MCM}
\end{figure}

\begin{lemma}
\label{lemmaforsteps}Let $p\neq q$ and $p+q=1$. Then, as $m$ tends to
infinity, the expected number of steps needed for $\mathcal{M}$ to reach $%
v_{m}$, given that we start at $v_{1}$, satisfies%
\begin{equation}
\mu _{1}=\left\{ 
\begin{array}{ll}
e^{m\ln {\frac{q}{p}}+o(m)} & \quad \text{if $p<q$} \\ 
\frac{m}{p-q}+o(m) & \quad \text{if $p>q$}%
\end{array}%
\right.
\end{equation}
\end{lemma}

\begin{proof}
For $i=0,1,\ldots ,m$, let $\mu _{i}$ denote the expected number of steps
needed to reach $v_{m}$. Clearly, $\mu _{m}=0$ and $\mu _{0}=1+\mu _{1}$.
Furthermore, for $i=1,\ldots m-1$, it follows that 
\begin{equation}
\mu _{i}=1+p\mu _{i+1}+q\mu _{i-1}
\end{equation}%
i.e.%
\begin{eqnarray}
\mu _{i+1}-\mu _{i} &=&\frac{q}{p}(\mu _{i}-\mu _{i-1})-\frac{1}{p}  \notag
\\
&=&\left( \frac{q}{p}\right) ^{i}(\mu _{1}-\mu _{0})-\frac{1}{p}%
\sum_{j=0}^{i-1}\left( \frac{q}{p}\right) ^{j} \\
&=&-\left( \frac{q}{p}\right) ^{i}-\frac{1}{q-p}\left( \left( \frac{q}{p}%
\right) ^{i}-1\right)  \notag
\end{eqnarray}%
Consequently, we have that%
\begin{eqnarray}
\sum_{i=1}^{m-1}[\mu _{i+1}-\mu _{i}] &=&-\mu _{1}\Leftrightarrow  \notag \\
\mu _{1} &=&\sum_{i=1}^{m-1}\left[ \left( 1+\frac{1}{q-p}\right) \left( 
\frac{q}{p}\right) ^{i}-\frac{1}{q-p}\right] \\
\mu _{1} &=&\left( 1+\frac{1}{q-p}\right) \frac{\left( \frac{q}{p}\right)
^{m}-\frac{q}{p}}{\frac{q}{p}-1}-\frac{m-1}{q-p}  \notag
\end{eqnarray}%
which concludes the proof of the lemma for large $m$.
\end{proof}

\medskip

Denote in the following by $\mathcal{M}_{1}$ the Markov chain of the
stochastic process defined in~\cite{Nowak05} (see Section~\ref%
{Nowak-model-subsec} for an overview), when the underlying graph is the 
clique-wheel $G_{n}$, cf.~Figure~\ref{axinoswithring}. 
The next definition will be useful for the discussion below.

\begin{definition}[Ring steps]
\label{ringstep} A transition of the Markov chain $\mathcal{M}_{1}$ is
called a \emph{ring step} if it results in a change of the number of black
vertices in the outer ring (i.e.~ring vertices).
\end{definition}

We now present some \emph{domination statements} that simplify the Markov
chain $\mathcal{M}_{1}$. More specifically, all these statements will
increase the probability of reaching fixation when we start with one black
ring vertex, such that we finally get an upper bound on $h_{ring}$.

\begin{description}
\item[$D_{1}$:] Let $v$ be a vertex on the outer ring, and let $%
v^{\prime }$ be its (unique) neighbor in the clique. We will \emph{forbid}
transitions of the Markov chain $\mathcal{M}_{1}$, where a white colored $%
v^{\prime }$ tries to make white a black colored~$v$.

\item[$D_{2}$:] Fixation is forced when either of the following happens:\vspace{-0.1cm}
\begin{description}
\item[$A_{1}$:] The outer ring reaches $\log {n}$ black vertices.\vspace{-0.1cm}
\item[$A_{2}$:] The number of ring steps in order to reach $\log {n}$ black
ring vertices is more than $\Theta (\log ^{2}{n})$.\vspace{-0.1cm}
\item[$A_{3}$:] The clique reaches $n$ black vertices.\vspace{-0.1cm}
\item[$A_{4}$:] A black clique vertex makes black a white vertex on the outer ring.
\end{description}
\end{description}

Let now $\mathcal{M}_{2}$ be the modified Markov chain after these
domination statements are imposed. Given these statements, we can prove the following. 

\begin{proposition} \label{propositionsinglearc}
If ${\cal M}_2$ has not yet reached an absorbing state, then the outer ring consists 
of a single arc of consecutive black vertices, or has no black vertex at all. 
\end{proposition}

\begin{proof}
This follows by noticing that a second black arc can only be created either
by a white colored vertex on the clique making white its black neighbor on
the outer ring (this is impossible because of $D_{1}$), or by a black vertex
on the clique making black its white neighbor on the outer ring (this will
lead to absorption at $A_{4}$ of $D_{2}$).
\end{proof}

\medskip

The following definitions will be useful in the sequence.

\begin{definition}[Offspring]
\label{offspring}Given the history $H_{t}$ of the Markov chain up to time $t$, 
the current state $S\subseteq V$ at time $t$ (i.e.~the set of black
colored vertices in the graph), and a vertex $v\in S$, we will say that $v$
is an \emph{offspring} of $u\in V$ if and only if there exists a transition
path in $H_{t}$ that proves that $v$ is black because of $u$.
\end{definition}

Notice in Definition~\ref{offspring} that $u$ is not necessarily black at
time $t$.

\begin{definition}[Birth in the clique]
\label{birth}We will say that a vertex $v^{\prime }$ is \emph{born} in the
clique if and only if its (unique) neighbor $v$ in the outer ring is black
and makes a transition to the clique.
\end{definition}

Notice in Definition~\ref{birth} that the color of $v^{\prime }$ is
irrelevant before $v^{\prime }$ is born in the clique. We only need that the
color of $v$ is black. Furthermore, the above definition allows for a
specific vertex to be born more than once (i.e.~at different time steps).
The proof of our main theorem can now be reduced to a collection of lemmas.
Lemma~\ref{lemmaring} concerns the behavior of the ring.

\begin{lemma}
\label{lemmaring} Let $\mathcal{B}_{1}$ be the stochastic process describing
the ring steps in Markov chain $\mathcal{M}_{2}$. Given that we do not have
absorption at $A_{4}$, then $\mathcal{B}_{1}$ is a birth-death process with
forward bias equal to $r$. Furthermore, given that we start with a single 
black vertex on the ring, the following hold:
\begin{enumerate}
\item[(1)] The probability that the number of black vertices in the outer
ring reaches $\log {n}$ before absorption at $A_{2},A_{3}$ or $A_{4}$ is at
most $\frac{1-\frac{1}{r}}{1-\left( \frac{1}{r}\right) ^{\log n}}$.
\item[(2)] The probability that more than $\log ^{2}{n}$ ring steps are
needed in order to reach $\log {n}$ black colored vertices in the ring, or
to reach absorption in $A_{2}$, $A_{3}$, or $A_{4}$ is at most $O\left( 
\frac{1}{\log {n}}\right) $.
\end{enumerate}
\end{lemma}

\begin{proof}
Recall that we do not allow transitions where the clique affects the number
of black colored vertices in the outer ring (by the domination statements $D_{1}$ and $A_{4}$). 
Then, it can be easily seen that the forward bias of the 
birth-death process $\mathcal{B}_{1}$ (i.e.~the ratio of the forward
probability over the backward probability) is $\frac{\frac{2r}{W}\frac{1}{3}%
}{\frac{2}{W}\frac{1}{3}}=r$, where $W$ is the sum of the fitness of every
vertex in the graph. Thus, part (1) of the lemma follows by equation~(\ref%
{rho-general}) (for an overview of birth-death processes, see also~\cite%
{NorrisMarkov,nowak06-book}).

For part (2), let $X$ denote the number of ring steps needed in order to
reach $\log {n}$ black colored vertices in the ring, or to reach absorption
in $A_{2},A_{3}$ or $A_{4}$. Then $X$ is stochastically dominated by the
number of steps needed for Markov chain $\mathcal{M}$ (cf.~Figure~\ref{MCM})
to reach $v_{m}$, with $m=\log {n}$ and $p=\frac{r}{r+1}$. Hence, by Lemma~%
\ref{lemmaforsteps} and Markov's inequality, we get that%
\begin{equation*}
\Pr (X\geq \log ^{2}{n})\leq c\frac{1}{\log {n}}
\end{equation*}%
for some positive constant $c=c(r)$.
\end{proof}

\medskip

The next lemma bounds the number of vertices that are born in the clique (see Definition~\ref{birth}).

\begin{lemma}
\label{lemmaforborn} Given that we start with a single 
black vertex on the ring, the probability that we have more than $\log^7{n}$
births in the clique is at most $O\left( \frac{1}{\log{n}} \right)$.
\end{lemma}

\begin{proof}
For the proof, we will ignore for the moment what happens in the clique and
how the clique affects the ring, since these steps are either forbidden (by $%
D_{1}$) or lead to absorption (by $A_{4}$).

Let $Y$ be the number of births in the clique (see Definition~\ref{birth})
that we observe between two ring steps. Notice that at any time before
absorption, there will be exactly 2 white colored vertices in the outer ring
that can perform a ring step (see Definition~\ref{ringstep}). Furthermore,
if the number of black vertices in the ring is more than 2, then not all
black vertices can affect the number of black vertices in the ring. We now
restrict ourselves, to observe only ring-involved moves (forgetting about
the clique), that is, transitions where only vertices of the ring that can
cause a ring step or a birth in the clique are chosen. Given that $\mathcal{M%
}_{2}$ (i.e.~the modified Markov chain) has not been absorbed, the
probability that a ring step happens next is 
\begin{equation*}
p_{step}=\frac{2(1+r)}{2+zr}\frac{1}{3}
\end{equation*}%
where $z$ is the number of black colored vertices in the outer ring.
Similarly, the probability that a birth in the clique happens next is%
\begin{equation*}
p_{birth}=\frac{zr}{2+zr}\frac{1}{3}
\end{equation*}%
Consequently, the random variable $Y+1$ is stochastically dominated by a
geometric random variable with probability of success%
\begin{equation*}
p=\frac{p_{step}}{p_{step}+p_{birth}}=\frac{2r+2}{zr+2r+2}\geq \frac{1}{\log 
{n}}
\end{equation*}%
where in the last inequality we used the observation that at any time before
absorption, the number of black vertices in the ring is at most $\log {n}$
because of $A_{1}$. But then, by Markov's inequality, we have that%
\begin{equation*}
\Pr (Y+1\geq \log ^{5}{n}+1)\leq \frac{\frac{1}{p}}{\log ^{5}{n}+1}\leq 
\frac{1}{\log ^{4}{n}}
\end{equation*}

But by part (2) of Lemma~\ref{lemmaring}, the probability that there are
more than $\log ^{7}{n}$ births in the clique before the Markov chain is
absorbed is by Boole's inequality at most%
\begin{equation*}
\log ^{2}{n}\Pr (Y\geq \log ^{5}{n})+O\left( \frac{1}{\log {n}}\right) \leq
O\left( \frac{1}{\log {n}}\right)
\end{equation*}%
which concludes the proof of the lemma.
\end{proof}

\medskip

The following lemma states that it is highly unlikely that the clique will
affect the outer ring, or that the number of black vertices in the clique
will reach $n$.

\begin{lemma}
\label{lemmaforA3A4} Given that we start with a single 
black vertex on the ring, the probability of absorption at $A_3$ or~$A_4$ is at
most $O\left( \frac{1}{\log{n}} \right)$.
\end{lemma}

\begin{proof}
For the purposes of the proof, we assign to each birth in the clique a
distinct label. Notice that, by Lemma~\ref{lemmaforborn}, we will use at
most $\log ^{7}{n}$ labels with probability at least~${1-O\left( \frac{1}{%
\log {n}}\right)}$. If we end up using more than $\log ^{7}{n}$ labels
(which happens with probability at most~${O\left( \frac{1}{\log {n}}\right)}$
by Lemma~\ref{lemmaforborn}), then we stop the process and assume that we
have reached one of the absorbing states. Furthermore, whenever a black
vertex $v$ in the clique with label $i$ replaces one of its neighbors with
an offspring, then the label of $v$ is inherited to its offspring.

In order for $\mathcal{M}_2$ to reach absorption at $A_3$, the clique must
have $n$ black vertices. Since each of these vertices has a label $j \in
[\log^7{n}]$, there exists at least one label $i$ such that at least $\frac{n%
}{\log^7{n}}$ vertices have label $i$. Similarly, if $\mathcal{M}_2$ reaches
absorption at $A_4$ and $v$ is the corresponding affected ring vertex, then
there exists a label $i$, such that $v$ has label $i$. We will call a label $%
i$ \emph{winner} if there are at least $\frac{n}{\log^7{n}}$ vertices in the
clique that have label $i$, or the outer ring is affected by a clique vertex
of label $i$. Clearly, if $\mathcal{M}_2$ reaches absorption at $A_3$ of $%
A_4 $, there must be at least one winner.

Recall that, by Corollary~\ref{interpretation}, the probability that a
single black vertex in the clique either reaches $\frac{n}{\log ^{7}{n}}$
offsprings or affects the outer ring is at most $\frac{7}{6n\left( \frac{4}{%
3r}-1\right) }+o\left( \frac{1}{n}\right) $. Consider now a particular label 
$i$. Then, if all the other black vertices of the graph that do not have
label $i$ (i.e.~black ring vertices or black clique vertices with label $%
j\neq i$) had fitness 1, then the probability that $i$ becomes a winner is
by Corollary~\ref{interpretation} at most $\frac{7}{6n\left( \frac{4}{3r}%
-1\right) }+o\left( \frac{1}{n}\right) $. The fact that the other black
vertices that do not have label $i$ have fitness $r$ can only reduce the
probability that $i$ becomes a winner. Therefore, considering all different
labels $i\in \lbrack \log ^{7}{n}]$ and using Boole's inequality, we
conclude that the probability to reach absorption at $A_{3}$ or $A_{4}$ is
at most 
\begin{equation*}
\log ^{7}{n}\left( \frac{7}{6n\left( \frac{4}{3r}-1\right) }+o\left( \frac{1%
}{n}\right) \right) +O\left( \frac{1}{\log {n}}\right) =O\left( \frac{1}{%
\log {n}}\right)
\end{equation*}%
where the term $O\left( \frac{1}{\log {n}}\right) $ in the left side
corresponds to the probability that we have more than $\log ^{7}{n}$ labels.
\end{proof}

\medskip

Finally, the following theorem concerns the probability of absorption of $\mathcal{M}_2$.

\begin{theorem}
\label{bound-M2-thm}For $n$ large, given that we start with a single black
vertex on the ring, the probability that $\mathcal{M}_{2}$ is absorbed at $%
A_{1}$ is at most $(1+o(1))\left( 1-\frac{1}{r}\right) $. Furthermore, the
probability of absorption at $A_{2}$, $A_{3}$, or $A_{4}$ is at most $%
O\left( \frac{1}{\log {n}}\right) $.
\end{theorem}

\begin{proof}
The bounds on the absorption at $A_{1}$ or $A_{2}$ follow from Lemma~\ref%
{lemmaring}, while the bounds on absorption at $A_{3}$ or $A_{4}$ follow
from Lemma~\ref{lemmaforA3A4}.
\end{proof}

\medskip

Recall now that $\mathcal{M}_{2}$ (the modified Markov chain) dominates $%
\mathcal{M}_{1}$ (the original Markov chain). Furthermore, recall that the
clique-wheel graph $G_{n}$ has $n$ clique vertices and $n$ ring vertices, and thus
the fixation probability of $G_{n}$ is $f_{G_{n}}=\frac{h_{clique}+h_{ring}}{%
2}$. Therefore, the next theorem is implied by Theorem~\ref{bound-M2-thm}
and Lemma~\ref{lemmaforhclique}.

\begin{theorem}
\label{bound-fixation-thm}For the Markov chain $\mathcal{M}_{1}$, and any $%
r\in \left( 1,\frac{4}{3}\right) $, $h_{ring}\leq (1+o(1))\left( 1-\frac{1}{r%
}\right) $. Therefore, as $n\rightarrow \infty $, the fixation probability
of the clique-wheel graph $G_{n}$ in Figure~\ref{axinoswithring} is by Lemma~%
\ref{lemmaforhclique}%
\begin{equation}
f_{G_{n}}\leq \frac{1}{2}\left( 1-\frac{1}{r}\right) +o(1)
\end{equation}
\end{theorem}

\section{Analysis of the aggregation model\label{parallel-model-sec}}

In this section, we provide analytic results on the new evolutionary
model of mutual influences. More specifically, in Section~\ref%
{potential-subsec} we prove that this model admits a \emph{potential function} 
for arbitrary undirected graphs and arbitrary initial fitness vector,
which implies that the corresponding dynamic system converges to a stable
state. Furthermore, in Section~\ref{clique-mutual-subsec} we prove fast
convergence of the dynamic system for the case of a complete graph, as well
as we provide almost tight upper and lower bounds on the limit fitness, to
which the system converges.

\subsection{Potential and convergence in general undirected graphs\label{potential-subsec}}

In the following theorem we prove convergence of the new model of mutual
influences using a potential function.

\begin{theorem}
\label{undir-conv-thm}Let $G=(V,E)$ be a connected undirected graph. Let $%
r(0)$ be an initial fitness vector of~$G$, and let $r_{\min }$ and $r_{\max}$ 
be the smallest and the greatest initial fitness in $r(0)$,
respectively. Then, in the model of mutual influences, the fitness vector $%
r(k)$ converges to a vector $[r_{0},r_{0},\ldots ,r_{0}]^{T}$ as~$%
k\rightarrow \infty $, for some value $r_{0}\in \lbrack r_{\min },r_{\max }]$.
\end{theorem}

\begin{proof}
Denote the vertices of $G$ by $V=\{u_{1},u_{2},\ldots ,u_{n}\}$. Let $k\geq
0 $. Then (\ref{P-exact}) implies that for any $i=1,2,\ldots ,n$, the
element $r_{u_{i}}(k+1)$ of the vector $r(k+1)$ is 
\begin{eqnarray*}
r_{u_{i}}(k+1) &=&\frac{1}{\Sigma (k)}\sum_{u_{j}\in N(u_{i})}\frac{%
r_{u_{j}}(k)}{\deg (u_{j})}\cdot r_{u_{j}}(k)+(1-\frac{1}{\Sigma (k)}%
\sum_{u_{j}\in N(u_{i})}\frac{r_{u_{j}}(k)}{\deg (u_{j})})\cdot r_{u_{i}}(k)
\\
&=&r_{u_{i}}(k)+\frac{1}{\Sigma (k)}\sum_{u_{j}\in
N(u_{i})}r_{u_{j}}(k)\cdot \frac{r_{u_{j}}(k)-r_{u_{i}}(k)}{\deg (u_{j})}
\end{eqnarray*}%
and thus%
\begin{equation}
\frac{r_{u_{i}}(k+1)}{\deg (u_{i})}=\frac{r_{u_{i}}(k)}{\deg (u_{i})}+\frac{1%
}{\Sigma (k)}\sum_{u_{j}\in N(u_{i})}r_{u_{j}}(k)\cdot \frac{%
r_{u_{j}}(k)-r_{u_{i}}(k)}{\deg (u_{i})\deg (u_{j})}
\label{ri-undir-explicit-2}
\end{equation}%
Therefore, by summing up the equations in (\ref{ri-undir-explicit-2}) for
every $i=1,2,\ldots ,n$ it follows that%
\begin{eqnarray}
\sum_{u_{i}\in V}\frac{r_{u_{i}}(k+1)}{\deg (u_{i})} &=&\sum_{u_{i}\in V}%
\frac{r_{u_{i}}(k)}{\deg (u_{i})}+\frac{1}{\Sigma (k)}\sum_{u_{i}u_{j}\in E}%
\frac{(r_{u_{j}}(k)-r_{u_{i}}(k))^{2}}{\deg (u_{i})\deg (u_{j})}
\label{ri-undir-explicit-3} \\
&\geq &\sum_{u_{i}\in V}\frac{r_{u_{i}}(k)}{\deg (u_{i})}  \notag
\end{eqnarray}%
Define now the potential function $\phi (k)=\sum_{u_{i}\in V}\frac{%
r_{u_{i}}(k)}{\deg (u_{i})}$ for every iteration $k\geq 0$ of the process.
Note by Observation~\ref{convex-bound-parallel-obs} that $\Sigma
(k)=\sum_{u_{i}\in V}r_{u_{i}}(k)\leq nr_{\max }$ is a trivial upper bound
for $\Sigma (k)$. Therefore, (\ref{ri-undir-explicit-3}) implies that%
\begin{eqnarray}
\phi (k+1)-\phi (k) &=&\frac{1}{\Sigma (k)}\sum_{u_{i}u_{j}\in E}\frac{%
(r_{u_{j}}(k)-r_{u_{i}}(k))^{2}}{\deg (u_{i})\deg (u_{j})}  \label{Delta-phi-1} \\
&\geq &\frac{1}{nr_{\max }}\sum_{u_{i}u_{j}\in E}\frac{%
(r_{u_{j}}(k)-r_{u_{i}}(k))^{2}}{\deg (u_{i})\deg (u_{j})} 
\ > \ \frac{1}{n^{3}r_{\max }}\sum_{u_{i}u_{j}\in E}(r_{u_{j}}(k)-r_{u_{i}}(k))^{2} \notag
\end{eqnarray}%
Furthermore, note that $r_{\max }\cdot \sum_{u_{i}\in V}\frac{1}{\deg (u_{i})%
}\leq nr_{\max }$ is a trivial upper bound for $\phi (k)$. Therefore, since $%
\phi (k+1)\geq \phi (k)$ for every $k\geq 0$ by (\ref{ri-undir-explicit-3}),
it follows that $\phi (k)$ converges to some value $\phi _{0}$ as $%
k\rightarrow \infty $, where $\phi (0)\leq \phi _{0}\leq nr_{\max }$.
Consider now an arbitrary $\varepsilon >0$ and let $\varepsilon ^{\prime }=%
\frac{\varepsilon ^{2}}{n^{3}r_{\max }}$. Then, since $\phi (k)%
\stackbin[k\rightarrow \infty]{}{\longrightarrow}\phi _{0}$, there exists $%
k_{0}\in \mathbb{N}$, such that $|\phi (k+1)-\phi (k)|<\varepsilon ^{\prime
} $ for every $k\geq k_{0}$. Therefore, (\ref{Delta-phi-1}) implies that for
every edge $u_{i}u_{j}\in E$ of $G$ and for every $k\geq k_{0}$,\vspace{-0.2cm}%
\begin{eqnarray}
(r_{u_{j}}(k)-r_{u_{i}}(k))^{2} &\leq &\sum_{u_{p}u_{q}\in
E}(r_{u_{p}}(k)-r_{u_{q}}(k))^{2}  \label{ri-undir-explicit-4} \\
&\leq &n^{3}r_{\max}\cdot|\phi (k+1)-\phi(k)| \ \leq \ n^{3}r_{\max }\cdot \varepsilon^{\prime} \ = \ \varepsilon^{2} \notag
\end{eqnarray}%
Thus, for every $\varepsilon >0$, there exists $k_{0}\in \mathbb{N}$, such
that $|r_{u_{j}}(k)-r_{u_{i}}(k)|<\varepsilon $ for every $k\geq k_{0}$ and
for every edge $u_{i}u_{j}\in E$ of $G$. Therefore, since $G$ is assumed to
be connected, all values $r_{u}(k)$, where $u\in V$, converge to the same
value $r_{0}$ as $k\rightarrow \infty $. Furthermore, since $r_{u}(k)\in
\lbrack r_{\min },r_{\max }]$ by Observation~\ref{convex-bound-parallel-obs}, 
it follows that $r_{0}\in \lbrack r_{\min },r_{\max }]$ as well. This
completes the proof of the theorem.
\end{proof}

\subsection{Analysis of the complete graph\label{clique-mutual-subsec}}

The next theorem provides an almost tight analysis for the limit fitness
value $r_{0}$ and the convergence time to this value, in the case of a
complete graph (i.e.~a homogeneous population).

\begin{theorem}
\label{r0-clique-thm}Let $G=(V,E)$ be the complete graph with $n$ vertices
and $\varepsilon >0$. Let $\alpha \in [0,1]$ be the portion of
initially introduced mutants with relative fitness $r\geq 1$ in $G$, and let 
$r_{0}$ be the limit fitness of $G$. Then $|r_{u}(k)-r_{v}(k)|<\varepsilon $
for every $u,v\in V$, when\vspace{-0.1cm}%
\begin{equation}
k\geq (n-2)\cdot \ln (\frac{r-1}{\varepsilon })\vspace{-0.1cm} \notag
\end{equation}%
Furthermore, for the limit fitness $r_{0}$,\vspace{-0.1cm}%
\begin{equation}
r_{0}\leq 1+\alpha (r-1)+\frac{\alpha (1-\alpha )}{1+\alpha (r-1)}\cdot 
\frac{(r-1)^{2}}{2}  \label{r0-upper-statement}
\end{equation}%
and
\begin{eqnarray}
r_{0} &\geq &\frac{1+\alpha (r-1)+\sqrt{(1+\alpha (r-1))^{2}+2\alpha
(1-\alpha )(r-1)^{2}}}{2}  \label{r0-lower-statement} \\
&\geq &1+\alpha (r-1)  \notag
\end{eqnarray}
\end{theorem}

\begin{proof}
Since $G$ is symmetric, we do not distinguish among the different placements 
$S\subseteq V$ of the $\alpha n$ initially introduced mutants. Furthermore,
at every iteration $k\geq 0$, there exist by symmetry two different
fitnesses $r_{1}(k)$ and $r_{2}(k)$ for the vertices of $S$ and of $%
V\setminus S$, respectively. Thus, it suffices to compute only $r_{1}(k)$
and $r_{2}(k)$ for every $k\geq 0$. Let $\Delta (k)=r_{1}(k)-r_{2}(k)$.
Then, $\Delta (0)=r-1$. It follows now by (\ref{recursion-P}) and (\ref%
{P-clique}) that for every $k\geq 0$%
\begin{eqnarray}
r_{1}(k+1) &=&(1-\frac{(1-\alpha )nr_{2}(k)}{(n-1)\Sigma (k)})\cdot r_{1}(k)+%
\frac{(1-\alpha )nr_{2}(k)}{(n-1)\Sigma (k)}\cdot r_{2}(k)  \label{r1-clique}
\\
&=&r_{1}(k)-\Delta (k)\frac{(1-\alpha )nr_{2}(k)}{(n-1)\Sigma (k)}  \notag
\end{eqnarray}%
Similarly, 
\begin{eqnarray}
r_{2}(k+1) &=&\frac{\alpha nr_{1}(k)}{(n-1)\Sigma (k)}\cdot r_{1}(k)+(1-%
\frac{\alpha nr_{1}(k)}{(n-1)\Sigma (k)})\cdot r_{2}(k)  \label{r2-clique} \\
&=&r_{2}(k)+\Delta (k)\frac{\alpha nr_{1}(k)}{(n-1)\Sigma (k)}  \notag
\end{eqnarray}%
where $\Sigma (k)=\alpha nr_{1}(k)+(1-\alpha )nr_{2}(k)$. Subtracting now (%
\ref{r2-clique}) from (\ref{r1-clique}), it follows that 
\begin{eqnarray*}
\Delta (k+1) &=&\Delta (k)-\Delta (k)\cdot \frac{\Sigma (k)}{(n-1)\cdot
\Sigma (k)} \\
&=&\Delta (k)\frac{n-2}{n-1}
\end{eqnarray*}%
and thus, since $\Delta (0)=r-1$, it follows that for every $k\geq 0$%
\begin{equation}
\Delta (k)=(r-1)\cdot \left( \frac{n-2}{n-1}\right) ^{k}  \label{D(k)-exact}
\end{equation}%
Therefore, in particular, $\Delta (k)>0$ for every $k\geq 0$ if and only if $%
r>1$. Let now $\varepsilon >0$ be arbitrary. Then $|\Delta (k)|\leq
\varepsilon $ if and only if 
\begin{eqnarray}
\left( \frac{n-2}{n-1}\right) ^{k} &\leq &\frac{\varepsilon }{r-1}%
\Leftrightarrow  \notag \\
\left( 1+\frac{1}{n-2}\right) ^{k} &\geq &\frac{r-1}{\varepsilon }
\label{D(k)-condition}
\end{eqnarray}%
However, $\left( 1+\frac{1}{n-2}\right) ^{n-2}\rightarrow e$ as $%
n\rightarrow \infty $. Thus, for sufficiently large $n$, (\ref%
{D(k)-condition}) is satisfied when $e^{\frac{k}{n-2}}\geq \frac{r-1}{%
\varepsilon }$, or equivalently when 
\begin{equation}
k\geq (n-2)\cdot \ln (\frac{r-1}{\varepsilon })  \label{k-condition}
\end{equation}

Recall by Theorem~\ref{undir-conv-thm} that $r_{1}(k)\rightarrow r_{0}$ and $%
r_{2}(k)\rightarrow r_{0}$ for some value $r_{0}$, as $k\rightarrow \infty $%
, and thus also $\alpha r_{1}(k)+(1-\alpha )r_{2}(k)\rightarrow r_{0}$ as $%
k\rightarrow \infty $. Furthermore, it follows by (\ref{r1-clique}) and (\ref%
{r2-clique}) that%
\begin{equation}
\alpha r_{1}(k+1)+(1-\alpha )r_{2}(k+1)=\alpha r_{1}(k)+(1-\alpha )r_{2}(k)+%
\frac{\alpha (1-\alpha )}{(\alpha r_{1}(k)+(1-\alpha )r_{u_{2}}(k))}\cdot 
\frac{\Delta ^{2}(k)}{n-1}  \label{weighted-sum-clique}
\end{equation}%
That is, $\alpha r_{1}(k)+(1-\alpha )r_{2}(k)$ is a non-decreasing function
of $k$, and thus $\alpha r_{1}(k)+(1-\alpha )r_{2}(k)\geq \alpha r+(1-\alpha
)$. Therefore, for every $k\geq 0$,%
\begin{equation}
\alpha r_{1}(k)+(1-\alpha )r_{2}(k)\leq 1+\alpha (r-1)+\frac{\alpha
(1-\alpha )}{1+\alpha (r-1)}\cdot \frac{1}{n-1}\sum_{k=0}^{\infty }\Delta
^{2}(k)  \label{weighted-sum-clique-upper-1}
\end{equation}%
The sum $\sum_{k=0}^{\infty }\Delta ^{2}(k)$ can be computed by (\ref%
{D(k)-exact}) as%
\begin{equation}
\sum_{k=0}^{\infty }\Delta ^{2}(k)=(r-1)^{2}\cdot \frac{1}{1-(\frac{n-2}{n-1}%
)^{2}}=(r-1)^{2}\frac{(n-1)^{2}}{2n-3}  \label{square-sum}
\end{equation}%
Substituting now (\ref{square-sum}) into (\ref{weighted-sum-clique-upper-1}%
), it follows that%
\begin{equation}
\alpha r_{1}(k)+(1-\alpha )r_{2}(k)\leq 1+\alpha (r-1)+\frac{\alpha
(1-\alpha )}{1+\alpha (r-1)}\cdot (r-1)^{2}\frac{n-1}{2n-3}
\label{weighted-sum-clique-upper-2}
\end{equation}%
Therefore, since $\frac{n-1}{2n-3}\rightarrow \frac{1}{2}$ as $n\rightarrow
\infty $, and since $\alpha r_{1}(k)+(1-\alpha )r_{2}(k)\rightarrow r_{0}$
as $k\rightarrow \infty $, it follows by (\ref{weighted-sum-clique-upper-2})
that for sufficiently large $n$ and $k$, 
\begin{equation}
r_{0}\leq 1+\alpha (r-1)+\frac{\alpha (1-\alpha )}{1+\alpha (r-1)}\cdot 
\frac{(r-1)^{2}}{2}  \label{r0-clique-upper-bound}
\end{equation}

Recall by (\ref{weighted-sum-clique}) that $\alpha r_{1}(k)+(1-\alpha
)r_{2}(k)$ is non-decreasing on $k$, and thus $\alpha r_{1}(k)+(1-\alpha
)r_{2}(k)\leq r_{0}$. Therefore, it follows by (\ref{weighted-sum-clique})
and (\ref{square-sum}) that for every $k\geq 0$,%
\begin{equation*}
\alpha r_{1}(k)+(1-\alpha )r_{2}(k)\geq 1+\alpha (r-1)+\frac{\alpha
(1-\alpha )}{r_{0}}\cdot (r-1)^{2}\frac{n-1}{2n-3}
\end{equation*}%
Thus, since $\frac{n-1}{2n-3}\rightarrow \frac{1}{2}$ as $n\rightarrow
\infty $ and $\alpha r_{1}(k)+(1-\alpha )r_{2}(k)\rightarrow r_{0}$ as $%
k\rightarrow \infty $, it follows similarly to the above that for
sufficiently large $n$ and $k$, 
\begin{equation*}
r_{0}\geq 1+\alpha (r-1)+\frac{\alpha (1-\alpha )}{r_{0}}\cdot \frac{%
(r-1)^{2}}{2}
\end{equation*}%
and thus%
\begin{equation}
r_{0}^{2}-r_{0}(1+\alpha (r-1))-\frac{\alpha (1-\alpha )(r-1)^{2}}{2}\geq 0
\label{r0-clique-lower-bound-2}
\end{equation}%
Therefore, since $r_{0}>0$, it follows by solving the trinomial in (\ref%
{r0-clique-lower-bound-2}) that 
\begin{equation}
r_{0}\geq \frac{1+\alpha (r-1)+\sqrt{(1+\alpha (r-1))^{2}+2\alpha (1-\alpha
)(r-1)^{2}}}{2}  \label{r0-clique-lower-bound-3}
\end{equation}%
The statement of the theorem follows now by (\ref{k-condition}), (\ref%
{r0-clique-upper-bound}), and (\ref{r0-clique-lower-bound-3}).
\end{proof}

\medskip

The next corollary follows by Theorem~\ref{r0-clique-thm}. 

\begin{corollary}
\label{r0-clique-one-mutant-cor}Let ${G=(V,E)}$ be the complete graph with $n$
vertices. Suppose that initially \emph{exactly~one} mutant with relative fitness ${r\geq 1}$
is placed in $G$ and let $r_{0}$ be the limit fitness of $G$. 
Then ${1+\frac{r-1}{n}\leq r_{0}\leq 1+\frac{r^{2}-1}{2n}}$.
\end{corollary}

\begin{proof}
Since we have initially one mutant, it follows that $\alpha =\frac{1}{n}$.
Then, substituting this value of $\alpha $ in (\ref{r0-lower-statement}), we
obtain the lower bound $r_{0}\geq 1+\frac{r-1}{n}$. For the upper bound of $%
r_{0}$, it follows by substituting $\alpha $ in (\ref{r0-upper-statement})
that%
\begin{eqnarray}
r_{0} &\leq &1+\frac{r-1}{n}+\frac{\frac{1}{n}\frac{n-1}{n}}{\frac{r}{n}+(1-%
\frac{1}{n})}\cdot \frac{(r-1)^{2}}{2}  \notag \\
&=&1+\frac{r-1}{n}(1+\frac{n-1}{r+(n-1)}\cdot \frac{r-1}{2})
\label{r0-upper-clique-one-mutant} \\
&\leq &1+\frac{r-1}{n}(1+\frac{r-1}{2})  \notag \\
&=&1+\frac{r^{2}-1}{2n}  \notag
\end{eqnarray}%
\end{proof}

\section{Invasion control mechanisms\label{antibiotics-sec}} 

As stated in the introduction of this paper, our new evolutionary model of
mutual influences can be used to model control mechanisms over invading populations in networks. 
We demonstrate this by presenting two alternative scenarios in 
Sections~\ref{antibiotics-phases-subsec} and~\ref{antibiotics-continuous-subsec}. 
In both considered scenarios, we assume that $\alpha n$ individuals of relative fitness $r$ 
(the rest being of fitness $1$) are introduced in the complete graph with $n$ vertices. 
Then, as the process evolves, we periodically choose (arbitrarily) a small fraction $\beta\in[0,1]$ of individuals 
in the current population and we reduce their current fitnesses to a value that is 
considered to correspond to the healthy state of the system (without loss of generality, 
this value in our setting is $1$). In the remainder of this section, we call these modified 
individuals as ``stabilizers'', as they help the population resist to the invasion of the mutants.

\subsection{Control of invasion in phases\label{antibiotics-phases-subsec}}

In the first scenario of controlling the invasion of advantageous mutants in networks, 
we insert stabilizers to the population in phases, as follows. 
In each phase $k\geq 1$, we let the process evolve until all fitnesses~$\{r_{v}\ |\ v\in V\}$ become $%
\varepsilon $-relatively-close to their fixed point~$r_{0}^{(k)}$
(i.e.~until they\emph{\ }$\varepsilon $\emph{-approximate}~$r_{0}^{(k)}$). 
That is, until ${\frac{|r_{v}-r_{0}^{(k)}|}{r_{0}^{(k)}} < \varepsilon}$ for
every ${v\in V}$. Note by Theorem~\ref{undir-conv-thm} that, at every phase, 
the~fitness values always $\varepsilon $-approximate such a limit fitness $r_{0}^{(k)}$. 
After the end of each phase, we introduce $\beta n$ stabilizers, where $\beta \in \lbrack 0,1]$. 
That is, we replace~$\beta n$ vertices (arbitrarily chosen) by individuals of fitness $1$, i.e.~by resident individuals. 
Clearly, the more the number of phases, the closer the fixed point at the end of each phase will be to $1$. 
In the following theorem we bound the number of phases needed until the system stabilizes, 
i.e.~until the fitness of \emph{every} vertex becomes sufficiently close to $1$.

\begin{theorem}
\label{var1-clique-thm}
Let $G=(V,E)$ be the complete graph with $n$ vertices. Let $\alpha \in \lbrack 0,1]$ 
be the portion of initially introduced mutants with relative fitness $r\geq 1$ in $G$ 
and let $\beta \in \lbrack 0,1]$ be the portion of the stabilizers introduced at every phase. 
Let $r_{0}^{(k)}$ be the limit fitness after phase $k$ and let $\varepsilon,\delta >0$, 
be such that $\frac{\beta}{2} > \sqrt{\varepsilon}$ and $\delta > \frac{4}{3}\sqrt{\varepsilon}$. 
Finally, let each phase $k$ run until the fitnesses $\varepsilon$-approximate their fixed point $r_{0}^{(k)}$. 
Then, after

\begin{equation}
k\geq 1+ \frac{\ln{(\frac{\varepsilon +(1+\varepsilon )\frac{1+\alpha }{2}(r-1)}{\delta 
-\frac{4}{3}\sqrt{\varepsilon }})}}{\ln({\frac{1}{(1+\varepsilon )(1-\frac{\beta }{2})}})}
\end{equation}
phases, the relative fitness of \emph{every} vertex $u\in V$ is at most $1+\delta $.
\end{theorem}

\begin{proof} Consider the first phase, where initially there exist $\alpha n$ mutants with relative fitness $r$ and $(1-\alpha )n$ resident individuals with fitness $1$ each. Then, since $r\geq 1$, it follows by (\ref{r0-upper-statement}) for the fixed point $r_{0}^{(1)}$ after the first phase that
\begin{eqnarray}
r_{0}^{(1)} &\leq & 1+\alpha (r-1)\cdot (1+\frac{(1-\alpha )(r-1)}{2(1+\alpha (r-1))})  \nonumber \\
& = & 1+\frac{\alpha (r-1)}{2} \cdot (1+\frac{1+(r-1)}{1+\alpha (r-1)}) \label{r0-clique-upper-bound-2} \\
& \leq & 1+\frac{\alpha (r-1)}{2}\cdot (1+\frac{1}{\alpha })  \nonumber
\end{eqnarray}
i.e. 
\begin{equation}
r_{0}^{(1)}\leq 1+\frac{1+\alpha }{2}(r-1) \label{r0-upper-antivirus-1}
\end{equation}
Suppose that we let each phase $k\geq 1$ run until the fitnesses $\varepsilon$-approximate their fixed point $r_{0}^{(k)}$. Note that, at the start of the process, $(1-\alpha )n$ vertices have fitness $1$ and $\alpha n$ vertices have fitness $r$. Similarly, before the $k$th phase starts, $\beta n$ vertices have fitness $1$ and $(1-\beta )n$ vertices have fitness at most $(1+\varepsilon )r_{0}^{(k-1)}$. Then, we obtain similarly to (\ref{r0-upper-antivirus-1}) that the fixed point $r_{0}^{(k)}$ at iteration $k$ is in the worst case
\begin{eqnarray}
r_{0}^{(k)} & \leq & 1+\frac{1+(1-\beta )}{2}((1+\varepsilon )r_{0}^{(k-1)}-1) \nonumber \\
& = & 1+(1-\frac{\beta }{2})((1+\varepsilon )r_{0}^{(k-1)}-1) \nonumber
\end{eqnarray}
Therefore
\begin{equation}
(1+\varepsilon) r_{0}^{(k)} \leq (1+\varepsilon) + (1 + \varepsilon)(1 - \frac{\beta}{2})((1+\varepsilon)r_{0}^{(k-1)}-1) \nonumber 
\end{equation}
and thus
\begin{equation}
(1+\varepsilon )r_{0}^{(k)}-1\leq \varepsilon +(1+\varepsilon )(1-\frac{\beta}{2})((1+\varepsilon)r_{0}^{(k-1)}-1) \nonumber 
\end{equation}
Let now $\lambda =(1+\varepsilon )(1-\frac{\beta }{2})$. Then the last inequality becomes 
\begin{equation}
(1+\varepsilon )r_{0}^{(k)}-1 \leq \varepsilon +\lambda ((1+\varepsilon)r_{0}^{(k-1)}-1) \nonumber 
\end{equation}
and by induction we have
\begin{eqnarray}
(1+\varepsilon )r_{0}^{(k)}-1 & \leq & \varepsilon \sum_{i=0}^{k-2} \lambda^i +\lambda ^{k-1}((1+\varepsilon )r_{0}^{(1)}-1) \nonumber \\
& = & \varepsilon \frac{1-\lambda ^{k-1}}{1-\lambda }+\lambda ^{k-1}((1+\varepsilon )r_{0}^{(1)}-1) \nonumber
\end{eqnarray}
Therefore, (\ref{r0-upper-antivirus-1}) implies that
\begin{equation}
(1+\varepsilon )r_{0}^{(k)}-1\leq \varepsilon \frac{1-\lambda ^{k-1}}{1-\lambda }+\lambda ^{k-1}(\varepsilon +(1+\varepsilon )\frac{1+\alpha }{2} (r-1))  \label{X}
\end{equation}
At the end of the $k$th phase, the relative fitness of each vertex is at most $(1+\varepsilon )r_{0}^{(k)}$. Now, in order to compute at least how many phases are needed to reach a relative fitness $(1+\varepsilon)r_{0}^{(k)}\leq 1+\delta $ for every vertex $u\in V$, it suffices by (\ref{X}) to compute the smallest value of $k$, such that
\begin{equation}
\varepsilon \frac{1-\lambda ^{k-1}}{1-\lambda }+\lambda ^{k-1}(\varepsilon+(1+\varepsilon )\frac{1+\alpha }{2}(r-1))\leq \delta  \label{Y}
\end{equation}
Recall now that $\sqrt{\varepsilon} < \frac{\beta}{2} \leq \frac{1}{2}$ by assumption. 
Therefore $\lambda =(1+\varepsilon )(1-\frac{\beta }{2})<(1+\varepsilon )(1-\sqrt{\varepsilon })$, i.e.~$\lambda <1$. Thus $1-\lambda ^{k-1}<1$ and it suffices from (\ref{Y}) to compute the smallest number $k$ for which
\begin{equation}
\frac{\varepsilon }{1-\lambda }+\lambda ^{k-1}(\varepsilon +(1+\varepsilon )%
\frac{1+\alpha }{2}(r-1))\leq \delta   \label{W}
\end{equation}
Note now that
\begin{eqnarray}
\frac{\varepsilon }{1-\lambda } & = &\frac{\varepsilon }{1-(1+\varepsilon )(1-\frac{\beta }{2})} \nonumber \\
& = &\frac{\varepsilon }{\frac{\beta }{2}(1+\varepsilon )-\varepsilon } \nonumber
\end{eqnarray}
Thus, since $\frac{\beta }{2}>\sqrt{\varepsilon }$ by assumption, it follows that 
\begin{equation}
\frac{\varepsilon }{1-\lambda }<\frac{\varepsilon }{\sqrt{\varepsilon } (1+\varepsilon )-\varepsilon }=\frac{\sqrt{\varepsilon }}{1+\varepsilon - \sqrt{\varepsilon }}  \label{Z}
\end{equation}
However $1+\varepsilon -\sqrt{\varepsilon }\geq \frac{3}{4}$ for every $\varepsilon \in (0,1)$, and thus it follows by (\ref{Z}) that $\frac{\varepsilon}{1-\lambda}<\frac{4}{3}\sqrt{\varepsilon }$. Therefore it suffices from (\ref{W}) to compute the smallest number $k$ for which
\begin{displaymath}
\frac{4}{3}\sqrt{\varepsilon }+\lambda ^{k-1}(\varepsilon +(1+\varepsilon ) \frac{1+\alpha }{2}(r-1))\leq \delta. 
\end{displaymath}
That is,
\begin{displaymath}
\lambda^{k-1}\leq \frac{\delta -\frac{4}{3}\sqrt{\varepsilon }}{\varepsilon +(1+\varepsilon )\frac{1+\alpha }{2}(r-1)}
\end{displaymath}
or equivalently
\begin{displaymath}
k\geq 1+ \frac{\ln{(\frac{\varepsilon +(1+\varepsilon )\frac{1+\alpha }{2}(r-1)}{\delta 
-\frac{4}{3}\sqrt{\varepsilon }})}}{\ln({\frac{1}{(1+\varepsilon )(1-\frac{\beta }{2})}})}
\end{displaymath}
This completes the proof of the theorem.
\end{proof}

\subsection{Continuous control of invasion\label{antibiotics-continuous-subsec}}

In this section we present another variation of controlling the invasion of advantageous mutants, using our new evolutionary model. 
In this variation, we do not proceed in phases; we rather introduce \emph{at every single iteration} 
of the process $\beta n$ stabilizers, where $\beta \in \lbrack 0,1]$ is a small portion of the individuals of the population. 
For simplicity of the presentation, we assume that at every iteration the $\beta n$ stabilizers with relative 
fitness $1$ are the same.

\begin{theorem}
\label{var2-clique-thm}Let $G=(V,E)$ be the complete graph with $n$
vertices. Let $\alpha \in \lbrack 0,1]$ be the portion of initially
introduced mutants with relative fitness $r\geq 1$ in $G$ and let $\beta \in
\lbrack 0,1]$ be the portion of the stabilizers introduced at every iteration. 
Then, for every $\delta >0$, after\vspace{-0.15cm}%
\begin{equation}
k\geq \frac{r}{\beta }(n-1)\cdot \ln (\frac{r-1}{\delta })\vspace{-0.15cm}
\end{equation}%
iterations, the relative fitness of \emph{every} vertex $u\in V$ is at most $1+\delta$.
\end{theorem}

\begin{proof}
Recall that we assumed for simplicity reasons that at every iteration the $%
\beta n$ individuals with relative fitness $1$ are the same. Note
furthermore that at very iteration $k$ we have by symmetry three different
fitnesses on the vertices: (a) the $\alpha n$ initial mutants with fitness $%
r_{1}(k)$, (b) the $\beta n$ ``stabilizers'' with fitness $1$, and (c) the rest $(1-\alpha -\beta )n$
individuals with fitness $r_{2}(k)$, where $1\leq r_{2}(k)\leq r_{1}(k)$ by
Observation~\ref{convex-bound-parallel-obs}. Note that $r_{2}(0)=1$. Let $%
\gamma =1-\alpha -\beta $. Then, we obtain similarly to (\ref{r1-clique})
and (\ref{r2-clique}) in the proof of Theorem~\ref{r0-clique-thm} that for
every $k\geq 0$%
\begin{eqnarray}
r_{1}(k+1) &=&(1-\frac{(\gamma r_{2}(k)+\beta )n}{(n-1)\Sigma (k)})\cdot
r_{1}(k)+\frac{\gamma r_{2}(k)n}{(n-1)\Sigma (k)}\cdot r_{2}(k)+\frac{\beta n%
}{(n-1)\Sigma (k)}  \label{r1-antivirus2} \\
&=&r_{1}(k)-\frac{1}{(n-1)\Sigma (k)}(\gamma
nr_{2}(k)(r_{1}(k)-r_{2}(k))+\beta n(r_{1}(k)-1))  \notag
\end{eqnarray}%
and 
\begin{eqnarray}
r_{2}(k+1) &=&\frac{\alpha r_{1}(k)n}{(n-1)\Sigma (k)}\cdot r_{1}(k)+(1-%
\frac{(\alpha r_{1}(k)+\beta )n}{(n-1)\Sigma (k)})\cdot r_{2}(k)+\frac{\beta
n}{(n-1)\Sigma (k)}  \label{r2-antivirus2} \\
&=&r_{2}(k)+\frac{1}{(n-1)\Sigma (k)}(\alpha
nr_{1}(k)(r_{1}(k)-r_{2}(k))-\beta n(r_{2}(k)-1))  \notag
\end{eqnarray}%
where $\Sigma (k)=n(\alpha r_{1}(k)+\gamma r_{2}(k)+\beta )$. %
It follows now by (\ref{r1-antivirus2}) and (\ref{r2-antivirus2}) that%
\begin{eqnarray*}
r_{1}(k+1)-r_{2}(k+1) &=&r_{1}(k)-r_{2}(k) \\
&&-\frac{(\alpha nr_{1}(k)+\gamma nr_{2}(k))(r_{1}(k)-r_{2}(k))+\beta
n(r_{1}(k)-r_{2}(k))}{(n-1)\Sigma (k)} \\
&=&r_{1}(k)-r_{2}(k)-\frac{\Sigma (k)(r_{1}(k)-r_{2}(k))}{(n-1)\Sigma (k)}
\end{eqnarray*}%
and thus%
\begin{equation*}
r_{1}(k+1)-r_{2}(k+1)=\frac{n-2}{n-1}(r_{1}(k)-r_{2}(k))
\end{equation*}%
Therefore, since $r_{2}(0)=1$ and $r_{1}(0)=r\geq 1$, it follows that for
every $k\geq 0$,%
\begin{equation}
r_{1}(k)-r_{2}(k)=(r-1)\cdot \left( \frac{n-2}{n-1}\right) ^{k}
\label{Delta-antivirus2-3}
\end{equation}%
By substitution of (\ref{Delta-antivirus2-3}) into (\ref{r1-antivirus2}) it follows that%
\begin{equation}
r_{1}(k+1)=r_{1}(k)-\frac{n}{(n-1)\Sigma (k)}(\gamma r_{2}(k)(r-1)\left( 
\frac{n-2}{n-1}\right) ^{k}+\beta (r_{1}(k)-1))  \label{r1-antivirus2-2}
\end{equation}%
Define now $\Delta (k)=r_{1}(k)-1$. Then, it follows by (\ref%
{r1-antivirus2-2}) that%
\begin{eqnarray}
\Delta (k+1) &=&\Delta (k)\cdot (1-\frac{\beta n}{(n-1)\Sigma (k)})-\frac{%
\gamma nr_{2}(k)}{(n-1)\Sigma (k)}(r-1)\left( \frac{n-2}{n-1}\right) ^{k}
\label{r1-antivirus2-3} \\
&<&\Delta (k)\cdot (1-\frac{\beta n}{(n-1)\Sigma (k)})  \notag
\end{eqnarray}%
Note now that $\frac{\beta n}{\Sigma (k)}\geq \frac{\beta }{r}$, and thus (%
\ref{r1-antivirus2-3}) implies that%
\begin{equation}
\Delta (k+1)\leq \Delta (k)\cdot \left( 1-\frac{\beta }{r(n-1)}\right)
\label{r1-antivirus2-4}
\end{equation}%
Denote now for the purposes of the proof $\lambda =1-\frac{\beta }{r(n-1)}=%
\frac{n-1-\frac{\beta }{r}}{n-1}$. Then, it follows by the system of
inequalities in (\ref{r1-antivirus2-4}) that for every $k\geq 0$%
\begin{eqnarray}
\Delta (k) &\leq &\Delta (0)\cdot \lambda ^{k}  \label{r1-antivirus2-5} \\
&=&(r-1)\cdot \lambda ^{k}  \notag
\end{eqnarray}%
In order to compute at least how many iterations are needed such that $%
r_{1}(k)\leq 1+\delta $, i.e.~$\Delta (k)\leq \delta $, it suffices by (\ref%
{r1-antivirus2-5}) to compute the smallest value of $k$, such that%
\begin{equation*}
(r-1)\cdot \lambda ^{k}\leq \delta
\end{equation*}%
i.e.%
\begin{eqnarray}
\frac{1}{\lambda ^{k}}=(\frac{n-1}{n-1-\frac{\beta }{r}})^{k} &\geq &\frac{%
r-1}{\delta }\Leftrightarrow  \label{r1-antivirus2-7} \\
(1+\frac{1}{\frac{r}{\beta }(n-1)-1})^{k} &\geq &\frac{r-1}{\delta }  \notag
\end{eqnarray}%
However, $(1+\frac{1}{\frac{r}{\beta }(n-1)-1})^{\frac{r}{\beta }%
(n-1)}\rightarrow e$ as $n\rightarrow \infty $. Thus, for sufficiently large 
$n$, (\ref{r1-antivirus2-7}) is satisfied when 
\begin{equation*}
e^{\frac{k}{\frac{r}{\beta }(n-1)}}\geq \frac{r-1}{\delta }
\end{equation*}%
or equivalently when 
\begin{equation*}
k\geq \frac{r}{\beta }(n-1)\cdot \ln (\frac{r-1}{\delta })
\end{equation*}%
This completes the proof of the theorem.
\end{proof}

\begin{observation}
\label{bound-variation-2-obs}
The bound in Theorem~\ref{var2-clique-thm} of the number of iterations needed to 
achieve everywhere a sufficiently small relative fitness is independent of the 
portion $\alpha \in \lbrack 0,1]$ of initially placed mutants in the graph. 
Instead, it depends only on the initial relative fitness $r$ of the mutants 
and on the portion $\beta \in \lbrack 0,1]$ of the vertices, 
to which we introduce the stabilizers. 
\end{observation}

\section{Concluding remarks\label{conclusions}}

In this paper we investigated alternative models for evolutionary dynamics
on graphs. In particular, we first considered the evolutionary model
proposed in~\cite{Nowak05}, where vertices of the graph correspond to
individuals of the population. 
We provided in this model generic upper and lower bounds of the fixation probability on a general graph $G$ 
and we presented the first class of undirected graphs (called clique-wheels) that act as suppressors of
selection. Specifically, we proved that the fixation probability of the
clique-wheel graphs is at most one half of the fixation probability of the
complete graph (i.e.~the homogeneous population) as the number of vertices
increases. An interesting open question in this model is whether there exist 
functions $f_{1}(r)$ and~$f_{2}(r)$ 
(independent of the size of the input graph), 
such that the fixation probability of every undirected graph $G$ lies between~$f_{1}(r)$ 
and~$f_{2}(r)$. Another line of future research
is to investigate the behavior of the model of~\cite{Nowak05} in~the case
where there are more than two types of individuals (\emph{aggressive} vs.~\emph{non-aggressive}) in the~graph.

As our main contribution, we introduced in this paper a new evolutionary
model based on mutual influences between individuals. In contrast to the
model presented in~\cite{Nowak05}, in this new model all individuals
interact \emph{simultaneously} and the result is a compromise
between aggressive and non-aggressive individuals. In other words, the
behavior of the individuals in our new model and in the model of~\cite%
{Nowak05} can be interpreted as an \emph{``aggregation''} vs.~an \emph{``all-or-nothing''} 
strategy, respectively. We prove that our new evolutionary model 
admits a potential function, which guarantees the convergence of the system for any graph 
topology and any initial fitnesses on the vertices
of the underlying graph. Furthermore, we provide almost tight bounds on the
limit fitness for the case of a complete graph, as well as a bound on the
number of steps needed to approximate the stable state. 
Finally, our new model appears to be useful also in the abstract modeling of new control mechanisms 
over invading populations in networks. 
As an example, we demonstrated its usefulness by analyzing the behavior of two alternative control approaches. 
Many interesting open questions lie ahead in our new model. For instance, what is the speed
of convergence and what is the limit fitness in arbitrary undirected graphs?
What happens if many types of individuals simultaneously interact at every iteration?

\vspace{0.3cm}
\noindent \textbf{Acknowledgment.} Paul G. Spirakis wishes to thank Josep Diaz, 
Leslie Ann Goldberg, and Maria Serna, for many inspiring discussions on the model of~\cite{Nowak05}.

{
\bibliographystyle{abbrv}
\bibliography{ref-evolution}

\begin{thebibliography}{10}

\bibitem{Aldous-online-book}
D.~Aldous and J.~Fill.
\newblock {\em Reversible Markov Chains and Random Walks on Graphs}.
\newblock Monograph in preparation. Available at
  http://www.stat.berkeley.edu/aldous/RWG/book.html.

\bibitem{Antal06}
T.~Antal and I.~Scheuring.
\newblock Fixation of strategies for an evolutionary game in finite
  populations.
\newblock {\em Bulletin of Mathematical Biology}, 68:1923--1944, 2006.

\bibitem{Berger01}
E.~Berger.
\newblock Dynamic monopolies of constant size.
\newblock {\em Journal of Combinatorial Theory, Series B}, 83:191--200, 2001.

\bibitem{Broom10}
M.~Broom, C.~Hadjichrysanthou, and J.~Rychtar.
\newblock Evolutionary games on graphs and the speed of the evolutionary
  process.
\newblock In {\em Proceedings of the Royal Society A}, volume 466, pages
  1327--1346, 2010.

\bibitem{Broom-two-results}
M.~Broom, C.~Hadjichrysanthou, and J.~Rychtar.
\newblock Two results on evolutionary processes on general non-directed graphs.
\newblock In {\em Proceedings of the Royal Society A}, volume 466, pages
  2795--2798, 2010.

\bibitem{Broom08}
M.~Broom and J.~Rychtar.
\newblock An analysis of the fixation probability of a mutant on special
  classes of non-directed graphs.
\newblock In {\em Proceedings of the Royal Society A}, volume 464, pages
  2609--2627, 2008.

\bibitem{Broom09}
M.~Broom, J.~Rychtar, and B.~Stadler.
\newblock Evolutionary dynamics on small order graphs.
\newblock {\em Journal of Interdisciplinary Mathematics}, 12:129--140, 2009.

\bibitem{Taylor04}
A.~S. Christine~Taylor, Drew~Fudenberg and M.~A. Nowak.
\newblock Evolutionary game dynamics in finite populations.
\newblock {\em Bulletin of Mathematical Biology}, 66(6):1621--1644, 2004.

\bibitem{Durrett88}
R.~Durrett.
\newblock {\em Lecture notes on particle systems and percolation}.
\newblock Wadsworth Publishing Company, 1988.

\bibitem{Gintis00}
H.~Gintis.
\newblock {\em Game theory evolving: A problem-centered introduction to
  modeling strategic interaction}.
\newblock Princeton University Press, 2000.

\bibitem{Evolution-Population98}
J.~Hofbauer and K.~Sigmund.
\newblock {\em Evolutionary games and population dynamics}.
\newblock Cambridge University Press, 1998.

\bibitem{Imhof05}
L.~A. Imhof.
\newblock The long-run behavior of the stochastic replicator dynamics.
\newblock {\em Annals of applied probability}, 15(1B):1019--1045, 2005.

\bibitem{Kandori93}
M.~Kandori, G.~J. Mailath, and R.~Rob.
\newblock Learning, mutation, and long run equilibria in games.
\newblock {\em Econometrica}, 61(1):29--56, 1993.

\bibitem{Karlin75}
S.~Karlin and H.~Taylor.
\newblock {\em A First Course in Stochastic Processes}.
\newblock NY: Academic Press, 2nd edition, 1975.

\bibitem{kempe05}
D.~Kempel, J.~Kleinberg, and E.~Tardos.
\newblock Influential nodes in a diffusion model for social networks.
\newblock In {\em Proceedings of the 32nd International Colloquium on Automata,
  Languages and Programming (ICALP)}, pages 1127--1138, 2005.

\bibitem{Nowak05}
E.~Lieberman, C.~Hauert, and M.~A. Nowak.
\newblock Evolutionary dynamics on graphs.
\newblock {\em Nature}, 433:312--316, 2005.

\bibitem{Liggett85}
T.~M. Liggett.
\newblock {\em Interacting Particle Systems}.
\newblock Springer-Verlag, 1985.

\bibitem{Moran58}
P.~Moran.
\newblock Random processes in genetics.
\newblock In {\em Proceedings of the Cambridge Philosophical Society},
  volume~54, pages 60--71, 1958.

\bibitem{Mossel07}
E.~Mossel and S.~Roch.
\newblock On the submodularity of influence in social networks.
\newblock In {\em Proceedings of the 39th annual ACM Symposium on Theory of
  Computing (STOC)}, pages 128--134, 2007.

\bibitem{NorrisMarkov}
J.~R. Norris.
\newblock {\em Markov Chains}.
\newblock Cambridge University Press, 1999.

\bibitem{nowak06-book}
M.~A. Nowak.
\newblock {\em Evolutionary Dynamics: Exploring the Equations of Life}.
\newblock Harvard University Press, 2006.

\bibitem{Ohtsuki06}
H.~Ohtsuki and M.~A. Nowak.
\newblock Evolutionary games on cycles.
\newblock In {\em Proceedings of the Royal Society B: Biological Sciences},
  volume 273, pages 2249--2256, 2006.

\bibitem{Sandholm11}
W.~H. Sandholm.
\newblock {\em Population games and evolutionary dynamics}.
\newblock MIT Press, 2011.

\bibitem{Taylor06}
C.~Taylor, Y.~Iwasa, and M.~A. Nowak.
\newblock A symmetry of fixation times in evoultionary dynamics.
\newblock {\em Journal of Theoretical Biology}, 243(2):245--251, 2006.

\bibitem{Traulsen09}
A.~Traulsen and C.~Hauert.
\newblock Stochastic evolutionary game dynamics.
\newblock In {\em Reviews of Nonlinear Dynamics and Complexity}, volume~2. NY:
  Wiley, 2008.

\bibitem{Weibull95}
J.~Weibull.
\newblock {\em Evolutionary game theory}.
\newblock MIT Press, 1995.

\end{thebibliography}
}

\end{document}